%% file: main-05-arxiv.tex
\documentclass[12pt]{article}

\usepackage[english]{babel}

\usepackage[margin=1in]{geometry}

\usepackage{amsmath}
\usepackage{graphicx}
\usepackage{amssymb}
\usepackage{amsthm}
\usepackage{csquotes}
\usepackage{caption}
\usepackage{amssymb}
\usepackage{amsfonts}
\usepackage{color}
\usepackage{xcolor}
\usepackage{stackrel}
\usepackage{mathtools}
\usepackage{comment}

\usepackage[ruled,linesnumbered,lined]{algorithm2e}
\usepackage[colorlinks=true, allcolors=blue]{hyperref}

\newtheorem{theorem}{Theorem}
\newtheorem{lemma}{Lemma}
\newtheorem{claim}{Claim}
\newtheorem{conjecture}{Conjecture}
\newtheorem{corollary}{Corollary}
\newtheorem{proposition}{Proposition}
\newtheorem{observation}{Observation}

\theoremstyle{definition}
\newtheorem{definition}{Definition}
\newtheorem{example}{Example}

\include{notaion}

\title{
A Reduction from Chores Allocation to Job Scheduling
\footnote{A one-page  extended abstract of this paper has been presented in the EC 2023 conference.}
}
\author{ 
Xin Huang\thanks{{\tt kidxshine@gmail.com,
huangxin@inf.kyushu-u.ac.jp}} 
\\
Technion, Haifa 32000, Israel
\\
Kyushu University, Fukuoka, Japan
\and
Erel Segal-Halevi\thanks{{\tt erelsgl@gmail.com}}
\\
Ariel University, Ariel 40700, Israel
}

\newcommand{\er}[1]{\textcolor{blue}{#1}}
\newcommand{\erel}[1]{\er{\emph{(Erel says: #1)}}}
\newcommand{\xin}[1]{\textcolor{brown}{\emph{Xin: #1}}}

\usepackage{cleveref}

\begin{document}
\maketitle

\begin{abstract}
We consider allocating indivisible chores 
among agents with different cost functions, such that all agents receive a cost of at most a constant factor times their maximin share.
The state-of-the-art was presented by
Huang and Lu \cite{DBLP:conf/sigecom/HuangL21}.
They presented a non-polynomial-time algorithm, called HFFD, that attains an $11/9$ approximation, and a polynomial-time algorithm that attains a $5/4$ approximation.

In this paper, we show that HFFD can be reduced  to an algorithm called MultiFit,
developed by Coffman, Garey and Johnson in 1978 \cite{DBLP:journals/siamcomp/CoffmanGJ78}, for makespan minimization in job scheduling.
Using this reduction, we prove that the approximation ratio of HFFD is in fact equal to that of MultiFit, which is known to be $13/11$ in general, $20/17$ for $4\leq n\leq 7$, and $15/13$ for $n=3$.

Moreover,   we develop an algorithm for $(13/11+\epsilon)$-maximin-share allocation for any $\epsilon>0$, with run-time polynomial in the problem size and $1/\epsilon$. 
For $n=3$, we can improve the algorithm to find a
$15/13$-maximin-share allocation with run-time polynomial in the problem size.
Thus, we have practical algorithms that attain the best known approximation to maximin-share chore allocation.
\end{abstract}

\section{Introduction}

In this work, we study the problem of fair chore allocation. In this problem, there are some $m$ chores that have to be performed collectively by some $n$ agents. Each agent may have a different cost for doing each chore. Formaly, each agent $i$ has a cost function $v_i$, that maps each chore to its cost for the agent. Our goal is to partition the chores among the $n$ agents such that the resulting allocation is \emph{fair}, considering the different agents' cost function. 

We focus on a well known fairness notion --- \emph{maximin share} \cite{budish2011combinatorial}. Informally, the maximin share of agent $i$ is the worst cost in the most even $n$-partition that agent $i$ can make according to her own cost function. 
Obviously, it may be impossible to guarantee to agent $i$ a cost smaller than her maximin share, since if all agents have the same cost function, at least one of them will get exactly the maximin share.
Moreover, in some instances, there is no allocation that could give every agent at most her maximin share~\cite{Aziz_Rauchecker_Schryen_Walsh_2017}. 
Based on this fact, people ask what is the best approximation ratio for the maximin share?%
\footnote{
In the context of chores, the term \emph{minimax share} is probably more appropriate: the minimum, over all allocations, of the maximum bundle cost. However, the term maximin share is common in the literature on fair chores allocation, so we stick with it.
}

The best known approximation ratio for this problem is $\frac{11}{9}$~\cite{DBLP:conf/sigecom/HuangL21}. 
The algorithm in that paper generalizes an algorithm called First Fit Decreasing (FFD), originally developed for the problem of \emph{bin packing}, to chores allocation.
In bin packing, every chore has the same cost in every bin it is put into; in chores allocation, chores may have different costs to different agents.
We call the generalized algorithm Heterogeneous First-Fit Decreasing (HFFD; see Algorithm \ref{alg-hffd} below). They applied HFFD to the job scheduling problem, which can be considered as a special case of chores allocation such that all cost function are the same. At the end of the paper, they point out that there may be more connections with the job scheduling problem.

Actually, in the previous literature, applying FFD to job scheduling problem has already been studied. The corresponding algorithm is MultiFit algorithm, which was first proposed by Coffman, Garey and Johnson~\cite{DBLP:journals/siamcomp/CoffmanGJ78}. In that paper, the authors proved an $\frac{11}{9}$ approximation ratio for job scheduling and $\frac{20}{17}$ lower bound. We notice that these two ratios are exactly the same as the ratios for chores in the paper \cite{DBLP:conf/sigecom/HuangL21}, though their proof ideas are different. Is this a coincidence? Or are there some deep connections between these problems?

In this work, we answer this question positively by presenting a formal reduction from chores allocation to job scheduling. We show that HFFD for chores and MultiFit for job scheduling share the same approximation ratio. This gives a tight bound $\frac{13}{11}$ for the HFFD algorithm approximation ratio.

\subsection{Our contribution}
We analyze the HFFD algorithm from the work~\cite{DBLP:conf/sigecom/HuangL21} by a reduction from HFFD  to FFD. 
This allows us to both improve the approximation ratio of MMS for chores, and strengthen the bond between the scheduling and the fair allocation literature. The reduction works as follows.
Suppose that there is an instance in which HFFD cannot attain an approximation ratio  $\alpha$ or better, for some constant $\alpha\ge\frac{8}{7}$. 
Then, we can construct an instance in which FFD cannot attain an approximation ratio $\alpha$ or better.

The first attempt of adapting the FFD algorithm from bin-packing to identical-machines scheduling is the MultiFit algorithm \cite{DBLP:journals/siamcomp/CoffmanGJ78}. MultiFit uses binary search to find the minimum bin size such that FFD algorithm can allocate all jobs.
A line of follow-up works \cite{DBLP:journals/siamcomp/Friesen84,DBLP:journals/jal/FriesenL86,yue1990exact,yue1992simple} found that the exact approximation ratio of this algorithm is $\frac{13}{11}\approx 1.182$.
The approximation ratio improves to 
$\frac{20}{17}\approx 1.176$ when $n\leq 7$, 
to 
$\frac{15}{13}\approx 1.15$ when $n\leq 3$,
and to 
$\frac{8}{7}\approx 1.14$ when $n= 2$,



Huang and Lu~\cite{DBLP:conf/sigecom/HuangL21} first proved that HFFD could attain an approximation ratio $\frac{11}{9}\approx 1.222$. We show that HFFD has the same approximation ratio as MultiFit algorithm. 

Below, we denote by $\alpha^{(n)}$ the approximation ratio of MultiFit for scheduling on $n$ identical machines.

\begin{theorem}[Informal]
\label{thm-approximation-intro}
For any $n\geq 2$, 
HFFD among $n$ agents guarantees to each agent at most $\alpha^{(n)}$ of their MMS.
\end{theorem}
\Cref{thm-approximation-intro} is proved in \Cref{sec:reduction}. Combined with the results for the MultiFit algorithm, we have a tight approximation ratio $\frac{13}{11}$ for HFFD algorithm \cite{yue1992simple}, $\frac{20}{17}$ when the number of agents is 
at most $7$, and $\frac{15}{13}$ 
when the number of agents is 
at most $3$.

As explained by 
\cite{DBLP:conf/sigecom/HuangL21},
it is impractical to directly use the PTAS of job scheduling to compute an approximate maximin share of each agent, and then apply HFFD. They give an efficient algorithm that guarantees only $\frac{5}{4}$ MMS. We develop an FPTAS using binary search in a clever way. 

\begin{theorem}[Informal]
\label{thm-fptas-intro}
When all costs are non-negative integers,
for any $n\geq 2$ and any $\epsilon>0$,
we can use HFFD and binary search to find a $(1+\epsilon)\alpha^{(n)}$-approximate maximin share allocation in time $O(\frac{1}{\epsilon}n^2m\log(V_{\max}))$, where $V_{\max}$ is the maximum cost of an agent to all chores.
\end{theorem}
\Cref{thm-fptas-intro} is proved in \Cref{sec:fptas}.

The reason we cannot get an exact $\alpha^{(n)}$ MMS approximation in polynomial time is related to a peculiar property of FFD called \emph{non-monotonicity}: on some instances, FFD might allocate all the chores when the bin is small, but fail to allocate all the chores when the bin is slightly larger.
All known examples for non-monotonicity use a bin size that is very close to 1. Therefore, we conjecture, that if the bin size is sufficiently large, FFD becomes monotone. 
We extend the monotonicity property of FFD to HFFD, and prove that it holds for $n=2$ and $n=3$. These results imply:
\begin{theorem}[Informal]
\label{thm-monotonicity-intro}
When all costs are non-negative integers, 
for $n=3$, we can use HFFD and binary search to find a $\frac{15}{13}$-approximate maximin share allocation in time $O(n^2m\log(V_{\max}))$.
\end{theorem}
\Cref{thm-monotonicity-intro} is proved in  \Cref{sec:monotonicity}.

Our paper conveys a more general message: there are many  techniques used for proving approximation ratios in scheduling problems, and we believe that many of them can be adapted to prove approximation ratios for fair allocation.

\subsection{Related Work}

\subsubsection*{Job scheduling}
Job scheduling was one of the first optimization problems that have been studied.
In the \emph{identical machines} setting, we are given a set of $n$ identical processors (machines), and a set of processes (jobs) with different run-times. A \emph{schedule} is an assignment of jobs to machines. The \emph{makespan} of a schedule is the completion time of the machine that finishes last. The \emph{makespan minimization} problem is the problem of finding a schedule with a smallest makespan. This problem is NP-hard by reduction from \textsc{Partition}, but there are efficient algorithms attaining constant-factor approximations \cite{graham1966bounds,graham1969bounds,hochbaum1987using,xin2017direct}, and theoretic PTAS-s \cite{hochbaum1987using,alon1998approximation}.

The makespan minimization problem can be interpreted in a different way. The ``machines'' represent human agents. The ``jobs'' are chores that must be completed by these agents. The ``run-time'' of each chore represents the negative utility that each agent gets from this chore: a job $j$ with run-time $v(j)$ represents a chore that is ``worth'' $-v(j)$ dollars (that is, people are willing to pay $v(j)$ dollars to get an exemption from this chore).
A schedule is an allocation of the chores among the agents.
The makespan of an allocation corresponds to the (negative) utility of the poorest agent --- the agent with the smallest utility. In the economics literature, this utility is called the \emph{egalitarian welfare} of the allocation.
The problem of makespan minimization corresponds to the well-known problem of \emph{egalitarian allocation} --- an allocation that maximizes the utility of the poorest agent \cite{rawls1971theory,Moulin2004Fair}

The job scheduling problem becomes more challenging when each job can have a different run-time on each machine. 
This corresponds to the setting called \emph{unrelated machines}.
In this setting (in contrast to the setting of identical machines), the makespan minimization problem cannot have a polynomial-time algorithm that attains an approximation factor better than $3/2$ unless P=NP \cite{lenstra1990approximation}. The best known approximation ratio attained by a polynomial-time algorithm is $2$ \cite{lenstra1990approximation,verschae2014configuration}.%
\footnote{
The problem has an FPTAS when the number of machines is constant \cite{woeginger2000does}.
}
Translated to the fair allocation setting, this means that, when agents have subjective valuations (different agents may assign different costs to the same chore), it is hard to approximate the optimal egalitarian welfare.

The literature on fair allocation has adopted a more modest goal: instead of finding the egalitarian-optimal allocation overall, let each agent $i$ find an egalitarian-optimal  allocation assuming that all $n$ agents have the \emph{same} valuation as himself (this is the maximin-share of agent $i$), and then try to give each agent an approximation of his own maximin-share.

\subsubsection*{Chores}
The relation between scheduling and fair allocation was already noted in several papers. For example, 
\cite{barman2017approximation} presented an algorithm for fair item allocation, that is very similar to the Longest Processing Time First (LPT) algorithm for scheduling. Interestingly, for chores allocation, their algorithm attains an approximation ratio of $(4n-1)/(3n)$ of the MMS, which is exactly the same ratio that is attained by LPT for makespan minimization \cite{graham1969bounds}.

Techniques from scheduling have also been used in the problem of \emph{optimal-MMS}, that is, finding the optimal approximation of MMS attainable in a specific instance \cite{aziz2016approximation,nguyen2017approximate,heinen2018approximation}.

There is a different way to approximate the maximin-share for chores. Instead of multiplying the MMS by a constant, we can pretend that there are fewer agents available for doing the chores. That is, for some $d<n$, let each agent $i$ find an egalitarian allocation among $d$ agents with the same valuation as $i$; then attempt to give each agent $i$ at least this egalitarian value.
Recently, \cite{hosseini2022ordinal} showed that, when $d\leq 3n/4$, such an allocation always exists. Moreover, when $d\leq 2n/3$, such an allocation can be found in polynomial time. Both results use techniques from \emph{bin packing} --- another optimization problem, that can be seen as a dual of makespan-minimization \cite{hochbaum1987using}.

\subsubsection*{Goods}
While minimizing the largest completion time corresponds to fair allocation of chores,
the opposite scheduling problem --- \emph{maximizing} the \emph{smallest} completion time --- corresponds to fair allocation of goods.
Finding an egalitarian-optimal allocation in this case is much harder. As far as we know, 
there are polynomial-time constant-factor approximation algorithms only for special cases, such as when each agent $i$ values each good at either $0$ or $v_i$ \cite{bansal2006santa,feige2008allocations,asadpour2012santa}; the best known approximation ratio for this case is $12.3$ \cite{annamalai2017combinatorial}.
For the general case, the approximation ratios depend on the number of agents or goods \cite{bezakova2005allocating,asadpour2010approximation,chakrabarty2009allocating}.
There are also exact algorithms for two agents using branch-and-bound \cite{dall2007allocate}, and for any number of agents using constraint programming \cite{bouveret2009computing}.

An allocation giving each agent his or her maximin share may not exist for goods, just like for chores \cite{kurokawa2018fair,feige2021tight}.
The approximation ratios for goods are worse: the best known approximation ratio is $3/4$ \cite{ghodsi2018fair,garg2020improved}.
For goods, the algorithm of \cite{barman2017approximation} attains an approximation ratio of $2n/(3n-1)$. Interestingly, this is worse than the approximation ratio of LPT for maximizing the minimum completion time, which is $(3n-1)/(4n-2)$ \cite{csirik1992exact}.

\subsection{Overview of the Technique}
Basically, we reduce the HFFD algorithm for chores allocation to the FFD algorithm for bin packing. The reduction is done in two steps.

1) In Section \ref{sec:FFabstraction}, we show how to simplify a multi-agent instance to a single-agent instance. To this aim, we characterize the allocations that can be generated by HFFD with multiple agents, based on the valuation of a single agent --- the last remaining agent.
We call such an allocation \emph{First Fit Valid}.

2) In the second step,
given a First Fit Valid allocation
which is a counter-example for the approximation ratio of HFFD,
we construct an instance for the FFD algorithm such that the approximation ratio and the unallocated chores are preserved. To do this, we look at a minimal counter-example for the HFFD algorithm, which we call a \emph{Tidy-Up tuple}. 
Then we modify the instance such that the resulting allocation is a possible output of FFD, which gives us a counter-example for FFD.
This part is in Section \ref{sec:reduction}.

\section{Preliminaries}\label{sec:pre}
\subsection{Agents and Chores}
There are $n$ agents and $m$ chores. Let us denote the set of agents as $\agents$ and the set of chores as $\items$. A \emph{bundle collection} $\allocs$  is an ordered partition of $S\subseteq\items$ into $n$ disjoint sets. If $S=\items$, we can also call it a \emph{complete} bundle collection. We call each set $\alloci\in\allocs$ a \emph{bundle}. 

An \emph{assignment} $\sigma$ is a permutation over $\{1,2,..,n\}$. An allocation $(\allocs,\sigma)$ is a pair of bundle collection and an assignment such that agent $i$ will get bundle $\alloci[\sigma(i)]$. 
An allocation is called \emph{complete} if $\allocs$ is a complete bundle collection. 
We need the assignment $\sigma$ here because the HFFD algorithm does not allocate bundles according to a fixed order, and the order of bundles is important.
In the convention of the literature, agent $i$ gets bundle $\alloci$ so that there is no need for an assignment. 
Starting at Section \ref{sec:FFabstraction}, we do not need the assignment $\sigma$ anymore as only one agent is considered. 

Each agent has a cost function $v_i:2^{\items}\rightarrow \mathbb{R}^{+}$. In this paper, we only consider  additive cost functions, i.e., $v_i(S)=\sum_{c\in S}\vai{c}$ for every $S\subseteq \items$.

\begin{definition}[Maximin share]
Given a chores set $\items$, a cost function $v$, and the number of bundles $n$, \emph{maximin share} is defined as $\text{MMS}(\items,v,n)=\min_{\allocs}\max_{\alloci[j]\in\allocs}\vai{\alloci[j]}$, where $\allocs$ ranges over all complete bundle collections over $\items$. 
\end{definition}

For each agent $i$, let $\text{MMS}_i:=\text{MMS}(\items,v_i,n)$ denote the maximin share of agent $i$.


\begin{definition}[Maximin share allocation]
Allocation $(\allocs,\sigma)$ is called a \emph{maximin-share allocation} if
$\forall i, \vai{\alloci[\sigma(i)]}\le \text{MMS}_i$.

It ls called an \emph{$\alpha$-maximin-share-allocation}, for some $\alpha>1$, if  $\forall i, \vai{\alloci[\sigma(i)]}\le \alpha\cdot \text{MMS}_i$.
\end{definition}

\begin{definition}[Identical-order preference (IDO) \cite{DBLP:conf/sigecom/HuangL21}]
An instance is called \emph{identical-order preference} if there is an ordering of all chores $c_1,c_2\dots, c_m$ such that, for any agent $i$, $\vai{c_j}\ge\vai{c_k}$ if $j\le k$,
that is, chores are ordered from largest (most costly) to smallest (least costly).
\end{definition}

Informally, IDO instances are 
the ``hardest'' for maximin share.
This is formalized by the following lemma
(similar lemmas were proved in the context of goods \cite{bouveret2016characterizing,barman2017approximation}).
\begin{lemma}[\cite{DBLP:conf/sigecom/HuangL21}]
Any instance $\inst$ can be transformed into an IDO instance $\inst'$, so that any $\alpha$-MMS-allocation for $\inst'$ can be transformed to an $\alpha$-MMS-allocation for $\inst$ (and both transformations can be done in polynomial time).
\end{lemma}
Therefore, the approximation factor of any algorithm is equal to its approximation factor on IDO instances.

Given any IDO instance, we will fix an underlying ordering $c_1,c_2\dots, c_m$ satisfying the above definition.  We call this ordering a \emph{universal ordering}.

\begin{definition}
Suppose we are given an IDO instance. 
Given a bundle $B\subseteq\items$ and a positive integer $p$, the notation $B[p]$ denotes the $p$-th largest chore in the bundle. If there are ties, we break ties  according to the universal ordering. 
\end{definition}

By the above definition, each chore $c\in B$ is uniquely identify by an integer $p$ such that $c=B[p]$.


\subsection{First Fit Decreasing and MultiFit}
The First Fit Decreasing (FFD) algorithm \cite{johnson1973near} accepts as input a list of chores $\items$, a single cost function $v$, a threshold $\tau$, and a fixed number of bins $n$.%
\footnote{
In the original description of FFD, the number of bins is not fixed in advance: new bins are added until all items are filled. For our purposes, we fix the number of bins: chores that do not fit remain unallocated.
}
It fills one bin at a time. For each bin, it loops over the remaining chores from large to small. If the current chore \emph{fits} into the current bin (that is, it can be added to the current bin such that its total cost remains at most $\tau$), then it is inserted; otherwise, it is skipped and the algorithm examines the next chore. Once no more chores fit into the current bin, it is closed, a new bin is opened, and the algorithm loops over the remaining chores from large to small again. %
\footnote{
The original paper describes FFD in a different way: it keeps several open bins at once; it loops once over all items in decreasing order of size; and puts the current item in the first (smallest-index) bin into which it fits. 
However, a moment's thought reveals that the two descriptions are equivalent.
}
After $n$ bins are filled, the remaining chores remain unallocated.
We denote the output of FFD by
$FFD(\items, v, \tau, n)$, or just 
$FFD(\items, v, \tau)$ when $n$ is clear.

Given the number of bins $n$,
The \emph{approximation ratio of FFD w.r.t. $n$}, denoted $\alpha^{(n)}$, is the smallest real number satisfying the following property.
For any chores set $\items$ and cost function $v$ with MMS$(\items,v,n)=\mu$, FFD$(\items,v, \tau^{(n)} \mu,n)$ allocates all chores.
In other words:
$\alpha^{(n)}$ is the smallest factor such that, if all chores can fit into $n$ bins of size $\mu$, then 
FFD  allocates all chores into $n$ bins of size $\alpha^{(n)} \cdot \mu$.
Table \ref{table:approxratio} shows the approximation ratios of FFD for every $n$, as mentioned by \cite{johnson1973near} (for $n\leq 7$) and \cite{yue1990exact} (for $n\geq 8$).
\begin{center}
    \begin{tabular}{|c|c|c|c|c|}
    \hline
        $n = $ & 2 & 3 & 4,5,6,7 & $\ge 8$\\
        \hline
       $\alpha^{(n)} = $ & $8/7$ & $15/13$ &  $20/17$  & $13/11$\\
       \hline
    \end{tabular}
    \captionof{table}{approximation ratios of FFD for job scheduling.}\label{table:approxratio}
\end{center}

The MultiFit algorithm \cite{DBLP:journals/siamcomp/CoffmanGJ78}
runs FFD multiple times with the same $n$ and different values of $\tau$,
aiming to find an $\tau$ for which $FFD(\items, v, \tau, n)$ succeeds in allocating all the chores.
It computes the threshold using binary search: pick a threshold and run FFD; if it succeeds, pick a lower threshold; if it fails, pick a higher threshold. 
MultiFit is an approximation algorithm for makespan minimization in identical machine scheduling: its approximation ratio is the same as that of FFD in Table \ref{table:approxratio}.

\section{From identical to different cost functions}
Huang and Lu \cite{DBLP:conf/sigecom/HuangL21} generalize the FFD algorithm 
from bins to \emph{agents},
where each agent $i \in \agents$ may have a different cost $v_i(c)$ to each item (chore) $c$.
Moreover, each agent  may have a different capacity, 
denoted by $h_i$. 
We call this generalization \emph{Heterogeneous First Fit Decreasing} (HFFD; Algorithm \ref{alg-hffd}). 
In each iteration $k$, HFFD loops over the remaining chores from large to small, by the universal ordering. If the current chore fits into the current bin according to the capacity $h_i$ and the cost function $v_i$ of at least one remaining agent $i\in\agents$, it is inserted; otherwise, it is skipped and the algorithm examines the next chore. Once no more chores fit into the current bin, it is closed and allocated to one of the remaining agents $i\in\agents$ for whom the last item fitted.
This goes on until all chores are allocated, or until there are no remaining agents. In the former case, we say that the algorithm \emph{succeeded}; in the latter case, it \emph{failed}, since there are remaining chores that cannot be allocated to any agent.
Our goal is to determine a threshold $h_i$ for each agent $i\in\agents$ such that
\begin{itemize}
\item The algorithm with thresholds $h_1,\ldots,h_n$ succeeds (allocates all chores);
\item The thresholds are not too high, that is, $h_i \leq (13/11)\cdot \text{MMS}_i$.
\end{itemize}

\begin{algorithm}
\KwIn{An IDO instance $\inst$, threshold values of agents $(h_1,\dots,h_n)$.}
\KwOut{An allocation $(\allocs,\sigma)$.}
\BlankLine
Let $R=\items$ and $T=\agents$ and $\allocs=(\emptyset,\dots,\emptyset)$\;
\For(\tcp*[f]{Loop to generate bundles}\label{inalg-loop}){$k=1$ \KwTo $n$}
{
	 $\alloci[k]=\emptyset$\tcp*[f]{Generate bundles from $\alloci[1]$ to $\alloci[n]$}\;
	\For(\tcp*[f]{Consider chores in $R$ from largest to smallest}){ $j=1$ \KwTo  $|R|$}
	{
		\If{$\exists i\in T, \vai{\alloci[k]\cup R[j]}\le h_i$}
		{
			$\alloci[k]=\alloci[k]\cup R[j]$; 
		}	
	}
	Let $i\in T$ be an agent  such that $\vai{\alloci[k]}\le h_i$\;\label{inalg-select}
	$\sigma(i)=k$\tcp*[f]{Agent $i$ will get the bundle $A_k$}\;
	$T=T\setminus i$\;
	$R=R\setminus \alloci$\;		
}
\Return {$(\allocs,\sigma)$}
\caption{Heterogeneous First Fit Decreasing (HFFD)}\label{alg-hffd}	
\end{algorithm}

We denote by $HFFD(\items,\overrightarrow{v}, \overrightarrow{\tau})$ the output of HFFD on chores
$\items$,
cost functions $v_1,\ldots,v_n$ and thresholds $h_1,\ldots,h_n$.

Note that $FFD(\items, v, \tau,n)$, as described in the preliminaries, is 
 a special case of $HFFD(\items,\overrightarrow{v}, \overrightarrow{\tau})$ in which there are $n$ cost functions all equal to $v$ and all thresholds are equal to $\tau$, that is, 
$HFFD(\items,\allowbreak  v,\ldots,v,\allowbreak \tau, \ldots, \tau)$ where there are $n$ times $v$ and $n$ times $\tau$.

Based on \cite{yue1990exact}, one could think that we could attain a $(13/11)$-MMS allocation by letting each agent $i$ compute a threshold using binary search on $v_i$. Unfortunately, this does not work: the individual thresholds are not guaranteed to work with heterogeneous cost functions. 
The following example is based on Example 5.2 in \cite{DBLP:conf/sigecom/HuangL21}.

\begin{example}
\label{exm:chores}
There are $n=4$ agents and $m=15$ chores, with two types valuation types:
\begin{center}
\begin{tabular}{|c|c|c|c|c|c|c|c|c|c|c|c|c|c|c|c|}
\hline
Chore \#: & 1 & 2 & 3 & 4 & 5 & 6 & 7 & 8 & 9 & 10 & 11  & 12 & 13 &14  & 15 \\
\hline
\hline
Type A: & 51 & 28 & 27 & 27 & 27 & 26 & 12 & 12 & 11 & 11 &11  & 11 & 11 &11  & 10 \\
\hline
Type B: & 51 & 28 & 27 & 27 & 27 & 24 & 21 & 20 & 10 & 10 & 10 & 9 & 9 & 9 & 9  \\
\hline
\end{tabular}
\end{center}
For each type, the MMS value is $75$, using the following partitions:
\begin{itemize}
\item For Type A, we have $P_1=\{c_1,c_7,c_8\}$ with cost $51 + 12 + 12 = 75$, 
$P_2 = \{c_2,c_3,c_9\}$ with cost $28 + 27 + 11 = 66$, $P_3 = \{c_4,c_5,c_{10},c_{15}\}$ with cost $27 + 27 + 11 + 10 = 75$, $P_4 = \{c_6,c_{11},c_{12},c_{13},c_{14}\}$ with cost  $26 + 11 + 11 + 11 + 11 = 70$.
\item For Type B, we have $P_1 = \{c_1,c_6\}$ with cost $51 + 24 = 75$, $P_2 = \{c_2,c_3,c_8\}$ with cost $28 + 27 + 20 = 75$, $P_3=\{c_4,c_5,c_7\}$ with cost $27 + 27 + 21 = 75$, $P_4=\{c_9,\ldots,c_{15}\}$ with cost $10 + 10 + 10 + 9 + 9 + 9 + 9 = 66$.
\end{itemize} 
In fact, FFD with bin size $75$ yields exactly the same allocations.
However, HFFD with three agents of type A and one agent of type B, where all thresholds are $75$, fills bundles in the following order:
\begin{itemize}
\item $A_1$: \{$c_1,c_6$\}. Its cost for type B is $51+24=75$ and its cost for type A is $51+26 = 77 > 75$, so it is taken by the type B agent. The other three bundles are taken by  type A agents:
\item $A_2$: \{$c_2,c_3,c_7$\} with cost 28 + 27 + 12 = 67.
\item $A_3$: \{$c_4,c_5,c_8$\} with cost 27 + 27 + 12  = 66.
\item $A_4$: \{$c_9,...,c_{14}$\} with cost 11 + 11 + 11 + 11 + 11 + 11= 66.
\end{itemize}
There is no room in any of these bundles for $c_{15}$, since its cost is $10$, so $c_{15}$ remains unallocated.

\qed
\end{example}

\section{Proof Overview}\label{sec:proofoverview}
Before we go into the details, let us have a bird's-eye view of the whole proof. Our goal is: Given an output of HFFD algorithm such that there is at least one chore unallocated, we want to construct an instance of job scheduling such that the output of FFD algorithm cannot allocate all chores and the bin size is the same as HFFD algorithm.  

We do this in two main steps: 
1) Transform a multiple-agents setting to single-agent setting;
2) Reduce a potentially invalid output of FFD to a valid output of FFD.

In the first step, we transform a  multiple-agents setting to a single-agent setting by just looking at the valuation of the last agent.
The idea is that the last agent is present throughout the entire allocation, and thus evaluates all $n$ bundles.
 However, when we only take one valuation into consideration, we lose a lot of information about the allocation. For example, it is possible that the threshold of the last agent is $10$, but one of the allocated bundles has three chores whose costs to the last agent are $\{6,7,3\}$; such a bundle cannot appear in a valid output of FFD with the valuation of the last agent.
Here is the challenge: if the bundles
are arbitrary 
in the view of the last agent, then we cannot construct a job scheduling instance. 
To overcome this difficulty, we propose a necessary condition of the output of HFFD algorithm in the view of the last agent. We call this condition First Fit Valid (FFV). The FFV condition also gives us a lot of structure information for the reduction.  

The second step is more complex, and requires
several sub-steps. 
First, we try to shrink the size of instance as much as possible, while keeping the condition that at least one chore remains unallocated. This is done by a procedure we call the Tidy-Up procedure. After that, we can work on a ``tidy'' instance that has several useful properties. One property is that each bundle contains at most one ``illegal'' chore (that is, there is one chore such that, if it is removed, then the total bundle cost is at most the last agent's threshold).
The next thing is that we reduce the cost of the first illegal chore and then put it back. With a careful choice of the reduced cost, we can make sure that 
all bundles from the first bundle to the bundle that contained the first illegal chore are valid (that is, contain no illegal chores). 
Then, by induction, we get an output of FFD with some chore unallocated at the end.

The diagram in Figure \ref{fig:proofflow} shows an overview for the whole proof.
\begin{figure}
    \centering
    \includegraphics[width=16cm]{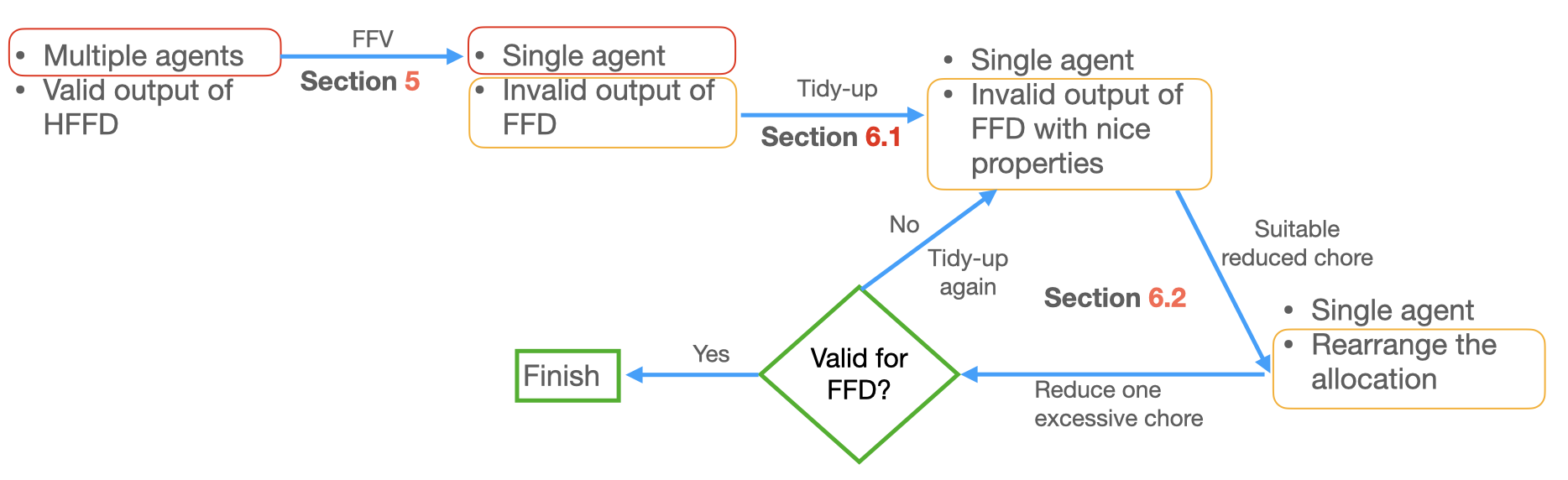}
    \caption{A bird's eye view of the whole proof.
    }
    \label{fig:proofflow}
\end{figure}

\section{First Fit Abstraction}\label{sec:FFabstraction}
In this section we 
reduce the multiple agents setting (chores allocation) to the single agent setting (job scheduling). 
In particular, we give a necessary condition for the output of the HFFD algorithm in terms of the cost function of the agent who gets the last bundle $\alloci[n]$. We denote this agent by $\last$, such that $\sigma(\last)=n$. We call this agent the \emph{last} agent. 

Given the necessary condition of the output of the HFFD algorithm
(proved in Proposition \ref{prop-HFFD-first-fit} below), we can forget all other agents and study the allocation satisfying this condition in the view of only one agent.
So from this section to the end, only the last agent is considered. Whenever we mention a cost function, it is the cost function of the last agent. In other words, there is a single cost function involved in all arguments. When we say ``an allocation'', there is no need for an assignment $\sigma$ anymore, so $\sigma$ is omitted from now on. 

We first introduce new notations to compare two bundles. Given a bundle $B$, let us recall the notation $B[p]$ is the $p$-th largest chore in bundle $B$. 
If there are several chores of the same cost, we choose  B[p] based on the universal ordering  (in Section \ref{sec:pre}).
We extend this notation such that, when $p>|B|$, we define $v(B[p])=0$.

\begin{definition}[Lexicographic order]
Given a cost function $v$ and two bundles $B_1,B_2$, we say $B_1\lex{=} B_2$ if $v(B_1[p])=v(B_2[p])$ for all $p$; $B_1\lex{\le} B_2$ if $B_1\lex{=}B_2$  or  $v(B_1[q])<v(B_2[q])$, where $q$ is the smallest index for which $v(B_1[q])\neq v(B_2[q])$.

Note that every two bundles are comparable by this order.

We define $B_1\lex{\ge}B_2$ if $B_2\lex{\le}B_1$; and $B_1\lex{<}B_2$ if $B_1\lex{\le}B_2$ and $B_1\lex{\neq}B_2$; and similarly define $\lex{>}$. 
\end{definition}

Note that, when we consider only a  single cost function, we can treat $\lex{=}$ as $=$. This is because, 
from the point of view of a single agent, chores with the same cost are equivalent, and can be exchanged anywhere, without affecting any of the arguments in the proofs.

With lexicographical order, we can define \emph{lexicographically maximal} subset. 
\begin{definition}[Lexicographically maximal subset]
Given a set of chores $S$, a cost function $v$ and a threshold $\tau$, let $S_{\leq \tau}=\{B\subseteq S\mid v(B)\le \tau\}$ be the set of all bundles with cost at most $\tau$. A bundle $B$ is a \emph{lexicographically maximal subset} of $S$, w.r.t. $v$ and $\tau$,
if $B\in S_{\leq \tau}$ and $B\lex{\ge} B'$ for any $B'\in S_{\leq \tau}$.
\end{definition}

Based on lexicographically maximality, we develop the following notions to relate the output of the FFD algorithm and HFFD algorithm.

\begin{definition}[Benchmark bundle]
Suppose that we are given a chores set $\items$, a bundle collection $\allocs$, a cost function $v$ and a threshold $\tau$. For each 
integer $k\in[n]$,
we define the \emph{$k$-th benchmark bundle} of $\allocs$, denoted $B_k$,
as a lexicographically maximal subset
of $\items\setminus \bigcup_{j<k} \alloci[j]$ w.r.t. $v$  and $\tau$.
\end{definition}

Intuitively,
after bundles 
$\alloci[1],\ldots,\alloci[k-1]$ 
have already been allocated,
the benchmark bundle is the lexicographically maximal feasible subset of the remaining items. We show below that this is exactly the outcome of running the 
FFD algorithm for generating the next bundle $\alloci[k]$. 

\begin{proposition}\label{prop-lexiffd}
Let $\mathbf{D}$ be the bundle collection which is output by FFD$\left(S,v,\tau\right)$. 
Then, for all $k\in[n]$, we have $D_k\lex{=}B_k$, 
where $B_k$ is the $k$-th benchmark bundle of $\mathbf{D}$.
\end{proposition}
\begin{proof}
By the description of FFD, we have $v(D_k)\le \tau$. By the definition of lexicographically maximal bundle, we have $B_k\lex{\ge}D_k$. To show equality, it is sufficient to prove $D_k\lex{\ge}B_k$ for all $i\in[n]$.
Suppose the contrary, and 
let $i$ be the smallest integer for which $D_k\lex{<}B_k$.
Let $q$ be the smallest integer for which $v(D_k[q])<v(B_k[q])$. 
When FFD constructs bundle $D_k$,
the chore $B_k[q]$ is available, it is processed before $D_k[q]$,
and the cost of the current bundle would not exceed $\tau$ if $B_k[q]$ is put into it at that moment, since $v(B_k)\leq \tau$. 
Therefore, FFD would put $B_k[q]$ into the current bundle before $D_k[q]$ --- a contradiction.
So we have $D_k\lex{\ge}B_k$, which gives us $D_k\lex{=}B_k$ as claimed.
\end{proof}

Next, we show that if we use HFFD to generate the bundle in some iteration $k$, then the bundle should be lexicographical greater than or equal to the $k$-th benchmark bundle. We call this property \emph{First Fit Valid}, and define it formally below.

\begin{definition}[First Fit Valid]
\label{def-first-fit-valid}
Given a bundle collection $\allocs$, a cost function $v$ and a threshold $\tau$, 
with $\tau\geq \text{MMS}(\items, v, n)$,
we call the tuple $(\items, \allocs,v, \tau)$ \emph{First Fit Valid} if for any $k\in[n]$, we have $A_k\lex{\ge}\bundle_k$, where $\bundle_k$ is the $k$-th benchmark bundle of $\allocs$,%
\footnote{
Note that $B_k$ is lexicographically-maximal among all subsets of $\bigcup_{j<k}A_j$ with cost at most $\tau$, 
but 
$A_k$ may be lexicographically larger than $B_k$ since its cost may be larger than $\tau$.
}
and the lexicographic comparison uses the cost function $v$.
\end{definition}

With the notion of First Fit Valid, we can characterize the output of HFFD algorithm as follows. 
Recall that $v_{\last},h_{\last}$ are the cost function and threshold of the last agent in HFFD.
\begin{proposition}\label{prop-HFFD-first-fit}
For any IDO instance $\inst$ and any threshold vector $(h_1,\ldots,h_n)$,
if Algorithm HFFD
produces a bundle collection $\allocs$,
then $(\items,\allocs,v_{\last},h_{\last})$ is a First Fit Valid  tuple. 
\end{proposition}
\begin{proof}
For all $k\geq 1$, 
we have to prove that $A_k\lex{\ge}B_k$, where $B_k$ is the $k$-th benchmark bundle of $\allocs$, where the lexicographic comparison uses the cost function $v_{\last}$.

if $A_k\lex{=}B_k$ we are done. 
Otherwise, let $q$ be the smallest index such that $\vai[\last]{A_k[q]}\neq \vai[\last]{B_k[q]}$.
Suppose for contradiction that 
$\vai[\last]{A_k[q]}<\vai[\last]{B_k[q]}$.
This implies that $B_k[q]$ precedes $A_k[q]$ in the universal ordering. 
In addition, if we add chore $B_k[q]$ to the current bundle before chore $A_k[q]$ is added, the cost of the current bundle would be acceptable by agent $\last$,
since $v_{\last}(B_k)\leq \tau$ by definition of a benchmark bundle. 
Therefore, HFFD would insert chore $B_k[q]$ to the current bundle before chore $A_k[q]$, which is a contradiction.
This implies that $\vai[\last]{A_k[q]}> \vai[\last]{B_k[q]}$.
Since $q$ is the first index in which $A_k$ and $B_k$ differ, 
$A_k\lex{>}B_k$.%
\end{proof}


We illustrate Proposition \ref{prop-HFFD-first-fit} on the HFFD run in Example \ref{exm:chores} with threshold $75$. 
Note that the last agent is of type A, 
so we use the type A cost function for the analysis.
For $k=1$, the benchmark bundle is the lexicographically-maximal subset of $\items$ with sum at most $75$, which is $B_1 = \{c_1,c_7,c_8\}$. 
Indeed, $A_1 = \{c_1,c_6\}$ is lexicographically larger than $B_1$.
For $k=2$, the benchmark bundle 
is the lexicographically-maximal subset of $\items\setminus \{c_1,c_6\}$ with sum at most $75$, which is $B_2 =\{c_2,c_3,c_7\}$. Indeed, $A_2 = B_2$. 
Similarly, it can be checked that $A_k=B_k$ for $k=3,4$.

The above propositions imply the following corollary,
showing the relationship between the First Fit Valid property and the FFD algorithm.
\begin{corollary}\label{cor-ffdsame}
Let $(\items,\allocs,v,\tau)$
be 
a First Fit Valid tuple, and $k\in[n]$ an integer such that  $v(\alloci[j])\le \tau$ for all $j\le k$.
Then 
the first $k$ bundles in FFD$(\items,v,\tau)$ are
lexicographically equivalent to
$(\alloci[1],\dots,\alloci[k])$.
\end{corollary}

\begin{proof}
By the property of First Fit Valid, for all $j\in[n]$ we have $\alloci[j]\lex{\ge}B_j$, where $B_j$ is the $j$-th benchmark bundle of $\allocs$. 
By the definition of benchmark bundle, $B_j$ is lexicographically maximal among subsets of remaining items with cost at most $\tau$. So $B_j\lex{\ge}\alloci[j]$ as $v(\alloci[j])\le \tau$. This implies that $\alloci[j]\lex{=}B_j$. By Proposition \ref{prop-lexiffd}, 
$B_j$ is the output of the $j$-th iteration of FFD$(\items, v,\tau)$. So $\alloci[j]$ is lexicographical equivalent to the $j$-th bundle  output by FFD.
\end{proof}


The following simple property of First Fit Valid tuples will be useful later.
\begin{lemma}
\label{lem-lex-larger-implies-cost-larger}

Let $(\items,\allocs,v,\tau)$ be a First Fit Valid tuple.
For any $k\in[n]$, 
and for any subset $C\subseteq \items \setminus 
\cup_{i<k} A_i$, 
if $C\lex{>}A_k$ then $v(C)>\tau$.
\end{lemma}
\begin{proof}
Suppose for contradiction that $v(C)\leq \tau$.
Since the $k$-th benchmark bundle $B_k$ is lexicographically maximal among the subsets of $\items \setminus 
\cup_{i<k} A_i$ with cost at most $\tau$, 
this implies 
$B_k \lex{\geq}C$.
By transitivity, 
$B_k \lex{>}A_k$.
But this contradicts the definition of First Fit Valid.
\end{proof}

We also show that any chore unallocated by HFFD cannot be too small. This was already proved in \cite{DBLP:conf/sigecom/HuangL21},
but we give a stand-alone proof below, using the First Fit Valid terminology. 
\begin{lemma}
\label{lem-nosmall}
Suppose that $(\items,\allocs, v,\tau)$ is First Fit Valid and $\mu := MMS(\items, v,n)$. 
If a chore $c^*$ is not in any bundle, that is $c^*\notin \bigcup_k\alloci[k]$, then $v(c^*)>\tau-\mu$.
\end{lemma}
\begin{proof}
Since the MMS is $\mu$, the sum of all chore costs is at most $n \mu$. By the pigeonhole principle, 
at least one of the bundles $A_k$ 
must have $v(A_k)\leq  \mu$.
Since $\alloci[k]\cup\{c^*\}\lex{>}\alloci[k]$,
Lemma \ref{lem-lex-larger-implies-cost-larger}
implies 
$v(\alloci[k]\cup\{c^*\})>\tau$.
By additivity, $v(c^*) > \tau-v(A_k) \geq  \tau-\mu$.
\end{proof}

\section{Reduction from First Fit Valid to FFD}
\label{sec:reduction}
In Section \ref{sec:FFabstraction}, we proved that First Fit Valid is a necessary condition for the output of HFFD algorithm. However, from the view of the last agent $\last$, 
in a First Fit Valid tuple 
$(\items, \allocs, v,\tau)$,
the allocation $\allocs$
might not be a possible output of FFD$(\items, v,\tau)$ (see Example \ref{exm:chores}).
In this section, we aim to ``reduce'' the First Fit Valid tuple to another First Fit Valid tuple in 
which the allocation \emph{is} a possible output of FFD, such that the unallocated chores and the threshold of the last agent remain the same. 

Intuitively, we focus on the view of the last agent $\last$ because the last agent participates in ``bidding" over all the previous $n-1$ bundles, so the costs of  all bundles are large for this agent.





\subsection{Tidy Up Tuple}

Previous work \cite{DBLP:journals/siamcomp/CoffmanGJ78} showed that a \emph{minimal instance} in which not all chores can be allocated has some nice  properties that can be used in the approximation ratio analysis. 
In a similar spirit, we show that, given a First Fit Valid tuple, we can construct a possibly smaller tuple with nice properties. We call this kind of tuple a \emph{tidy-up tuple}, 
and the procedure doing this construction the \emph{tidy-up procedure}.

Before we present the Tidy-Up Procedure, we introduce two related terms. They are helpful to identify redundant chores and bundles.
Intuitively, if a chore is allocated after the cost of the bundle reaches the maximin share, then we can remove it and this does not influence unallocated chores. These chores can be considered redundant. This motivates the following definition.

\begin{definition}
Given a cost function $v$, a threshold $\mu$ and a bundle $B$, a chore $c=B[p]$\footnote{Recall that each chore $c$ is uniquely identified by this notation and the universal ordering.} is called \emph{$\mu$-redundant} if $v\left(\bigcup_{q<p}B[q]\right)\ge \mu$
\end{definition}
To motivate this definition, suppose that  $B$ is a bundle constructed by FFD algorithm. A chore $B[q]$ is redundant if, before it was added to $B$, the total cost of $B$ was already at least the threshold $\mu$.


~~~
Next, we introduce the concept of \emph{domination} between sets of items. 
\begin{definition}\label{def-dominate}
Given a cost function $v$ and a set of chores $\items$, a set $A\subseteq\items$  \emph{dominates} a set $B\subseteq\items$ if there is a mapping $f:B\rightarrow A$ such that $v(f^{-1}(c'))\le v(c')$ for any chore $c'\in A$.
\end{definition}

Note that $f$ can be a many-to-one mapping. Here is an example (where the number in parentheses next to each chore is its cost). Let $A=\{c_0(1), c_1(2),c_2(6)\}$ and $B=\{c_3(2.5),c_4(3),c_5(1.5)\}$. The mapping $f:\{c_3\rightarrow c_2,c_4\rightarrow c_2, c_5\rightarrow c_1\}$ shows that $A$ dominates $B$, where $f^{-1}(c_0)=\emptyset,~ f^{-1}(c_1)=\{c_5\},~ f^{-1}(c_2)=\{c_3,c_4\}$.

The tidy-up procedure is described in Algorithm \ref{alg-tidyup}.
The procedure accepts as an input a First Fit Valid tuple with some unallocated chores, and simplifies it in several steps:
\begin{enumerate}
    \item Remove all unallocated chores except a single one. The rationale is that a single unallocated chore is sufficient for a counter-example.
    \item Remove all chores smaller than the single unallocated chore. The rationale is that these chores are processed after the unallocated chore, so removing them does not affect the unallocated chore.
    \item Remove all chores that are $\tau$-redundant, where $\tau$ is the threshold for First-Fit-Valid. This makes each bundle contain at most one chore that is ``illegal'' for an output of the FFD algorithm, which will simplify the future analysis.
    \item Finally, if any bundle in the original tuple dominates a bundle in the partition defining the MMS, then both bundles are removed, so the total number of bundles in both partitions drops by one. This last step is repeated until there are no more dominating bundles. 
\end{enumerate}
We emphasize that the purpose of the Tidy-Up Procedure is just to give a clear description of how the tidy up is done. Therefore, the run-time complexity of this procedure is not important. 

A tuple $(\items,\allocs, v,\tau)$ is called a \emph{Tidy-Up tuple} if it is the output of Tidy-Up Procedure on a First Fit Valid tuple.

\begin{algorithm}
\KwIn{A First Fit Valid tuple $(\items,\allocs, v,\tau)$ and an unallocated chore $c^* \in \items \setminus \cup_{k} A_k$.}
\KwOut{A Tidy Up tuple $(\items',\allocs',v,\tau)$.}
\BlankLine
\emph{Initialization}: 
Let $\mu := \text{MMS}(\items,v,n)$,
and let $\mathbf{P'}=(P'_1,\dots,P'_{n})$ be a partition of $\items$ such that $v(P'_i)\le \mu$ for all $i\in[n]$.
Let $\items':=\items$ and $\allocs':=\allocs$\;
Let $T:=\items'\setminus(\bigcup_{1\le i\le n}\alloci') = $ the set of unallocated chores in $(\items',\allocs',v,\tau)$\;
\tcp{When we remove a chore, it is removed from $\items'$, $\allocs'$ and $\mathbf{P'}$}
Remove all the chores in $T\setminus\{c^*\}$\;
Remove all chores $c$ with $v(c)<v(c^*)$\; 
Remove all $\mu$-redundant chores\;
\While{there exist $k,j$ such that $\alloci[k]'$ dominates $P_j'$}
{
    Let $f:P_j'\rightarrow\alloci[k]'$ be a domination mapping (Definition \ref{def-dominate})\;
    \For{$c'\in\alloci[k]'$}
    {
    Chore $c'$ belongs to some part in $\mathbf{P}'$; suppose that $c'\in P_r'$\;
    Modify $\mathbf{P}'$ by moving the chore $c'$ to $P_j'$ and moving the set  $f^{-1}(c')$ to $P_r'$\;
    \tcp{$f^{-1}(c')$ could be $\emptyset$}
    }
    \tcp{When the for loop completes, $\alloci[k]'=P_j'$.
    }
    Remove bundles $\alloci[k]'$ and $P_j'$\;
    \tcp{When we remove a bundle $\alloci[k]'$, the index of $\alloci[r]'$ becomes $r-1$ for all $r>k$. Also we remove  from $\items'$ all the chores in $\alloci[k]'$.}
}
\Return {$(\items',\allocs',v,\tau)$}
\caption{Tidy-Up Procedure}\label{alg-tidyup}	
\end{algorithm}

\begin{proposition}\label{prop-tidyup-FFV}
The outcome of the Tidy-Up Procedure is a First Fit Valid tuple. 
\end{proposition}
\begin{proof}
We prove that all four removals in the Tidy-Up Procedure preserve the First Fit Valid property. 

In all proofs below, we denote by $B_k$ the $k$-th benchmark bundle of the instance before the removal, and by $B_k'$ the $k$-th benchmark bundle of the instance after the removal.
By definition, $B_k$ is a lexicographically-maximal subset of  $\items\setminus \left(\bigcup_{i<k}A_i\right)$,
and $B_k'$ is a lexicographically-maximal subset of $\items'\setminus \left(\bigcup_{i<k}A_i\right)$,
where $\items' \subset \items$.
Since
$B_k'$ is lexicographically maximal in a smaller set than $B_k$, we have  $B_k\lex{\ge} B'_k$.

\paragraph{1. Removing unallocated chores.}
We prove that, given any First Fit Valid tuple $(\items,\allocs, v,\tau)$, if we remove any unallocated chore 
$c'\in \items\setminus\cup_i A_i$ 
from $\items$, the remainder  is First Fit Valid.
Let $B_k$ be the  $k$-th benchmark bundle according to $\items$ and $B_k'$ be the
$k$-th
benchmark bundle according to $\items\setminus c'$. 
We have 
$\alloci[k]\lex{\ge}B_k$ by definition of First Fit Valid.
When an unallocated chore is removed, any bundle $A_k$ does not change. Since $B_k\lex{\ge} B_k'$, by transitivity, 
$\alloci[k]\lex{\ge}B_k'$, so the tuple after the removal is  First Fit Valid.

\paragraph{2. Removing chores smaller than an unallocated chore.}
Let $c^*$ be an unallocated chore. We prove that, if we remove all chores smaller than $c^*$ from a First Fit Valid tuple
$(\items, \allocs, v, \tau)$,
the remaining tuple, which we denote by
$(\items', \allocs', v, \tau)$,
is still First Fit Valid. 

For any $k\in[n]$,
let $p$ be the smallest integer such that $v(\alloci[k][p+1])<v(c^*)$.
Then we have $\alloci[k]'=\bigcup_{q\le p}\alloci[k][q]$. 
~
Note that $\alloci[k]'\cup \{c^*\} \lex{>} \alloci[k]$ since $c^*$ is larger than all chores in $\alloci[k]\setminus \alloci[k]'$.
Since $c^*$ is unallocated, 
$\alloci[k]'\cup \{c^*\} \subseteq \items \setminus \cup_{i<k} A_i$.
Therefore, by Lemma \ref{lem-lex-larger-implies-cost-larger},
$v(\alloci[k]'\cup \{c^*\})>\tau$.

Let $B_k'$ be the $k$-th benchmark bundle of $\allocs'$. 
Assume for contradiction 
that $B_k'\lex{>}\alloci[k]'$, and consider two cases.
\begin{itemize}
\item Case 1: $\alloci[k]'$ is not a subset of $B_k'$.
This means that there is some $q\leq p$ for which $\alloci[k]'[q] \neq B_k'[q]$ and
$v(\alloci[k]'[q]) < v(B_k'[q])$.
But $\alloci[k]'[q]=\alloci[k][q]$,
so this implies 
$B_k'\lex{>}\alloci[k]$ too.
Since $B_k \lex{\geq} B_k'$, by transitivity 
$B_k \lex{>}\alloci[k]$, which contradicts that the original tuple is First Fit Valid.
\item Case 2: $\alloci[k]'$ is a strict subset of $B_k'$.
Let $c'$ be a chore in $B_k'\setminus \alloci[k]'$.
Note that $v(c') \geq v(c^*)$, as $c^*$ is the smallest chore after the removal.
So $v(B_k')\ge v(\alloci[k]'\cup \{c^*\})$.
By transitivity, 
$v(B_k')>\tau$, which contradicts the definition of $B_k'$ as a benchmark bundle.
\end{itemize}
Therefore, we must have $\alloci[k]'\lex{\ge}B_k'$, so the tuple after the removal is First Fit Valid.  

\paragraph{3. Removing redundant chores.}

Let $\mu := \text{MMS}(\items,v,n)$. We prove that if we remove all $\mu$-redundant chores
from a First Fit Valid tuple
$(\items, \allocs, v, \tau)$,
the remaining tuple 
$(\items', \allocs', v, \tau)$
is First Fit Valid. 

Let $\alloci[k]$ be some bundle with $\mu$-redundant chores. 
Let $p$ be the smallest integer such that $v\left(\bigcup_{q\le p}\alloci[k][q]\right)\ge \mu$. We have $\alloci[k]'=\bigcup_{q\le p}\alloci[k][q]$.  By Lemma \ref{lem-nosmall}, 
for any unallocated chore $c^*$,
$v(c^*)>\tau-\mu$, so $\left(\alloci[k]'\cup\{c^*\}\right)>\tau$. 
Now we just need to repeat the argument in paragraph 2 above to prove that the remaining tuple is First Fit Valid.

\paragraph{4. Removing entire bundles.}
We prove that, given any First Fit Valid tuple $(\items,\allocs, v,\tau)$, if we remove one bundle $A_k$, the remainder $(\items\setminus \alloci[k],~ (A_1,\ldots,A_{k-1},A_{k+1},\ldots, A_n),~ v,\tau)$ is First Fit Valid.
Let us not change the index of any bundle at this moment. For all $j>k$, the $j$-th benchmark bundle $B_j$ is the same before and after the removal, and we have $\alloci[j]\lex{\ge}B_j$ by definition of First Fit Valid.
For $j<k$, we have $B_j\lex{\ge} B'_j$, 
so by transitivity $\alloci[j]\lex{\ge}B_j\lex{\ge}B'_j$. Therefore, the tuple after the removal
is First Fit Valid. 
\end{proof}


Now, we analyze the properties of a Tidy-Up tuple with an unallocated chore.
\begin{lemma}[Tidy-Up Lemma]
\label{lem-tidyup}
Given a First Fit Valid tuple $(\items,\allocs, v,\tau)$ such that MMS$(\items,v,n)=\mu$ and $\tau\geq \mu$ and there is an unallocated chore $c^*$, the Tidy-Up tuple $(\items',\allocs', v,\tau)$:=Tidy-Up-Procedure$(\items,\allocs, v,\tau,c^*)$ has the following properties
(where $\mathbf{P}'$ is the partition at the end of Algorithm \ref{alg-tidyup}):
\begin{enumerate}
\item 

Subset properties:
\begin{enumerate}
\item \label{cond:subset} 

For each $k\in[n']$, there exists some $j\in [n]$ such that $\alloci[k]'\subseteq\alloci[j]$.

\item \label{cond:no-redundant}  
Every bundle $\alloci[k]'$ has no $\mu$-redundant chores for all $k\in[n']$.
\item \label{cond:smallest} Chore $c^*$ is in $\items'$, is the smallest chore in  $\items'$, and is the only chore 
not in any bundle $\alloci[k]'$.

\item 
\label{cond:hz} 
All chores in $\items'$ have a cost larger than $\tau-\mu$.

\item \label{cond:maxvalue} We have
$v(P_k')\le \mu$ for every bundle $P_k'$;
therefore, MMS$(\items',v,n')\le \mu$. 
\end{enumerate}

\item
No domination properties:
\begin{enumerate}
\item
\label{cond:atleasttwochores} 
Bundle $\alloci[k]'$ contains at least two chores
for all $k\in[n']$.
\item \label{cond:twochores} In every bundle $\alloci[k]'$, the sum of costs of the two largest chores is less than $\mu$.
 
\item \label{cond:threechores} Every bundle $P_k'$ contains at least 3 chores.

\end{enumerate}
\end{enumerate}
\end{lemma}
\begin{proof}
Condition \ref{cond:subset} holds since Tidy-Up Procedure only removes some chores or bundles from $\allocs$.
Condition \ref{cond:no-redundant} holds since Tidy-Up Procedure removes all $\mu$-redundant chores.

Condition \ref{cond:smallest} holds since Tidy-Up Procedure removes all chores smaller than $c^*$,
and all unallocated chores except $c^*$,
and does not remove the chore $c^*$ itself.
Condition \ref{cond:hz}
follows from Condition \ref{cond:smallest} (all chores are at least as large as $c^*$)
and \Cref{lem-nosmall}
($v(c^*)>\tau-\mu$).

Condition \ref{cond:maxvalue} holds since the Tidy-Up Procedure does not increase the cost of any bundle in $\mathbf{P}'$ which is not removed. So at the end of the algorithm, the partition $\mathbf{P}'$ has $n'$ bundles and the cost of each bundle is no larger than $\mu$. Therefore, MMS$(\items',v,n')\le \mu$.


To prove Condition \ref{cond:atleasttwochores}, suppose for contradiction that $|\alloci[k]'|=1$. Let $c_k\in \alloci[k]'$ be the only chore in the bundle. 
So
$\{c_k,c^*\}\lex{>}\alloci[k]'$.
By Proposition \ref{prop-tidyup-FFV} the tuple $(\items',\allocs',v,\tau)$ is First Fit Valid, 
so by Lemma \ref{lem-lex-larger-implies-cost-larger}
we have $v(\{c_k,c^*\})>\tau$.
Let $P_j'$ be the bundle in $\mathbf{P}'$ that contains $c_k$.
The bundle $P_j'$ cannot contain another chore, because $c^*$ is the smallest chore and $v(\{c_k,c^*\})>\tau\geq \mu$,
while all bundles in $\mathbf{P}'$ have a cost of at most $\mu$ by 
Condition \ref{cond:maxvalue}.
So $\alloci[k]'$ dominates $P_j'$. 
But Tidy-Up Procedure removes all bundles such that $\alloci[k]'$ dominates $P_j'$, which is a contradiction.


To prove Condition \ref{cond:twochores}, suppose for contradiction that $v(\{\alloci[k]'[1],\alloci[k]'[2]\})\geq \mu$.
Let $P_j'$ be the bundle in $\mathbf{P}'$ that contains $\alloci[k]'[1]$.
Construct a mapping $f:P_j'\rightarrow\alloci[k]'$ such that $f(\alloci[k]'[1])=\alloci[k]'[1]$ and $f(c)=\alloci[k]'[2]$ for all $c\neq\alloci[k]'[1]$, so 
$f^{-1}(\alloci[k]'[2])) = P_j'\setminus \alloci[k]'[1]$.
We have 
\begin{align*}
v(
f^{-1}(\alloci[k]'[2]))
& =
v(P_j') - v(\alloci[k]'[1])
\\
& \leq 
\mu - v(\alloci[k]'[1]) && \text{since $v(P_j')\leq \mu$ by Condition \ref{cond:maxvalue},}
\\
& \leq v(\alloci[k]'[2])  && \text{by the contradiction assumption,}
\end{align*}
so
the function $f$ shows that $\alloci[k]'$ dominates $P_j'$. This contradicts the operation of Tidy-Up Procedure.


To prove Condition \ref{cond:threechores}, suppose that there is a bundle $P_j'$ such that $|P_j'|\le 2$. If $|P_j'|=1$, then it is dominated by the bundle $\alloci[k]'$ such that $P_j'\subseteq\alloci[k]'$, which is impossible. 
Suppose that $|P_j'|=2$ and $P_j'=\{c_1,c_2\}$. 

Suppose first that one of these chores, say $c_2$, is the unallocated chore $c^*$. There is only one unallocated chore, so $c_1$ is allocated in some bundle, say $c_1\in\alloci[k_1]'$.
By condition \ref{cond:atleasttwochores}, $\alloci[k_1]'$ contains at least one other chore $c_3$ besides $c_1$. Since $c^*$ is the smallest chore, $v(c_3)\ge c^*$. Therefore, $\alloci[k_1]'$ dominates $P_j'$, which is impossible.

Next, suppose that both $c_1$ and $c_2$ are allocated. 
Suppose that $c_1\in\alloci[k_1]'$ and $c_2\in\alloci[k_2]'$. 
Without loss of generality, we assume that $k_1\le k_2$. 
Since $\allocs'$ is First Fit Valid,
$A_{k_1}'$ is (weakly) leximin larger than all bundles of cost at most $\tau$, that are subsets of $\items' \setminus \cup_{i<k_1} A_i$.
In particular, $P_j'$ is a subset of 
$\items' \setminus \cup_{i<k_1} A_i$, and its cost 
is at most $\mu$, which is at most $\tau$ by assumption. Therefore, $A_{k_1}'\lex{\ge} P_j'$.
Since one chore (namely $c_1$) is common to both $A_{k_1}'$ and $P_j'$, there must be at least one other chore in 
$A_{k_1}'$ with cost at least as large as $c_2$.
Therefore, $\alloci[k_1]'$ dominates $P_j'$, which is impossible. 
\end{proof}

Finally, we relate the Tidy-Up tuple to the output of FFD algorithm. 

\begin{definition}[Last chore]
\label{def-lastchore}
Given any bundle $B$, we call the chore $c\in B$ the \emph{last chore} of bundle $B$ if chore $c=B[q]$, where $q=|B|$. We call all other chores \emph{non-last} chore of bundle $B$. 
\end{definition}

\begin{definition}\label{def-excessive}
Given a Tidy-Up tuple $(\items, \allocs, v,\tau)$, a chore $c\in\alloci[j]$ is called \emph{excessive} if 
$v(\alloci[j])>\tau$
and $c$ is the last chore of bundle $\alloci[j]$.
\end{definition}

\begin{corollary}
In a Tidy-Up tuple $(\items, \allocs, v,\tau)$, if no chores are excessive, then $FFD(\items,v,\tau) = \allocs$.
\end{corollary}
\begin{proof}
``No excessive chores'' means that $v(A_i)\leq \tau$ for all $i\in[n]$. So the corollary follows directly from Corollary \ref{cor-ffdsame}.
\end{proof}
Since a Tidy-Up tuple has no 
$\mu$-redundant chores, removing an excessive chore brings the total bundle cost below $\mu$. 
Since $\mu\le \tau$, 
this implies $v(\alloci[j]\setminus \{c\})\leq \tau$.
This is a clue on how  to do the reduction: given an instance of HFFD, reduce the excessive chores one by one, and get an instance of FFD.

\subsection{The reduction}
\label{sub:reduction}

In this section, we show how to do the reduction from chores allocation to job scheduling. 
The reduction works only when, for every agent $i$, the threshold $h_i = \alpha\cdot \text{MMS}_i$, where $\alpha \geq 8/7$.
Given an instance of chores allocation such that HFFD with such thresholds cannot allocate all chores,
we construct a job-scheduling instance in which all jobs 
can fit into $n$ bins of size $MMS_{\last}$,
but FFD cannot allocate all jobs with bin capacity $h_{\last}$. 
This implies that if, for some $\alpha\ge\frac{8}{7}$, MultiFit attains an $\alpha$ approximation of the optimal  makespan, then HFFD attains an $\alpha$ approximation of the maximin share.

Our reduction is based on an induction proof. 
Given a Tidy-Up tuple, we  reduce the number of excessive chores by one at each step. To do this, we find 
the first bundle $A_k$ in the allocation that contains an excessive chore $c^*_k$, 
then we replace this chore by another chore $c^{\circ}_k$ with a smaller cost. 
We call this new chore  $c^{\circ}_k$ a \emph{reduced chore}. 
The cost of $c^{\circ}_k$ is determined such that, if it is allocated instead of the original chore $c^*_k$, it is not excessive any more. 
However, the reduced chore could be allocated to another bundle, and cause a butterfly effect on the whole allocation. 



To bound the  butterfly effect which is caused by the reduced chore, we require the reduced chore to have some special properties that we define below.

\noindent\fbox{\parbox{\textwidth}{%
\begin{definition}
\label{def-suitable-redu}
    {\bf Suitable reduced chore:} Let $(\items, \allocs, v, \tau)$ be a First Fit Valid tuple 
    with $\tau\geq \text{MMS}(\items, v, n)$, 
    with no $\tau$-redundant chores. 
    
    Let $k$ be the smallest integer such that there is an excessive chore $c_k^*\in\alloci[k]$ 
    (the smallest integer such that  $v(\alloci[k])>\tau$).
    Let $c_k^{\circ}$ be 
a chore with $v(c_k^{\circ})<v(c_k^*)$
and $\items'=\items\setminus\{c^*_k\}\cup\{c_k^{\circ}\}$. 
    
Let $\mathbf{D}$ be the output of $FFD(\items',v,\tau)$. The  chore $c_k^{\circ}$ is a \emph{suitable reduced chore} of $c_k^*$ if

 \begin{equation}
  \label{eq:suitable-redu}
     \bigcup_{j\le k}D_j=\left(\bigcup_{j\le k}A_j\right)\setminus\{c_k^*\}\cup\{c_k^{\circ}\}.
  \end{equation}
  
\end{definition}

}}

Intuitively, if a suitable reduced chore exists, then after replacing the first excessive chore $c_k^*$  with a suitable reduced chore $c_k^{\circ}$, FFD  allocates the same set of chores to the first $k$ bundles.
Comparing the bundle collection $(D_1,\dots,D_k,\alloci[k+1],\dots,\alloci[n])$  with bundle collection $\allocs$, the number of excessive chores decreases by at least 1. 
How to find a suitable reduced chore is the key to the whole reduction. 
\begin{lemma}
	\label{lem-suitable-reduced-chore}
	Let $(\items, \allocs, v, \tau)$ be a Tidy-Up tuple such that $\mu := \text{MMS}(\items,v,n)$ and 
	$\tau \ge \frac{8}{7} \mu$.
	If there is an excessive chore, then we can construct a suitable reduced chore. 
\end{lemma}

The proof is in Appendix \ref{app:reduced-chore}.
Based on Lemma \ref{lem-suitable-reduced-chore}, 
the following Algorithm \ref{alg-reduction} describes the procedure of removing all excessive chores from a First Fit Valid tuple. 
While there is an excessive chore in the tuple, we take the first one (by the order of bundles in the tuple), and replace it with a suitable reduced chore. 
After that, we know that the first bundles in the tuple (up to and including the bundle with the suitable reduced chore) do not contain any excessive chore, so we can replace them by the bundles that would be generated by FFD. 

At the end of the whole reduction, we get a First Fit Valid tuple that it is an output of FFD algorithm, with the same threshold and with an unallocated chore.


\begin{algorithm}
\KwIn{A First Fit Valid tuple $(\items,\allocs, v, \tau)$ with an unallocated chore $c^*$ }
\KwOut{A First Fit Valid tuple $(\items',\allocs',v,\tau)$ with an unallocated chore $c^*$ and no excessive chores.}
 \BlankLine
$(\items^*,\allocs^*,v, \tau)$=Tidy-Up$(\items,\allocs, v,\tau,c^*)$\;
    \While{there exists an excessive chore in $\allocs^*$}
    {
        Let $k$ be the smallest index such that $\alloci[k]^*$ has an excessive chore\;
        Let $c_k^*$ be the excessive chore in bundle $\alloci[k]^*$\;
        Let $c_k^{\circ}$ be a suitable reduced chore of $c_k^*$ (which exists by Lemma \ref{lem-suitable-reduced-chore})\;\label{inalg-suitablereduce}
        Update $\items':=\items^*\setminus\{c_k^*\}\cup\{c_k^{\circ}\}$\;
        Let $(D_1',\dots,D_k')$ be the first $k$ bundles output by $FFD(\items',v,\tau)$\;
        Let $n^*=|\allocs^*|$\;
        Update $\allocs':=(D_1',\dots,D_k',\alloci[k+1]^*,\dots,\alloci[n^*]^*)$\;
		$(\items^*,\allocs^*,v,\tau)$=Tidy-Up$(\items',\allocs', v,\tau,c^*)$\;
    }
	\Return {$(\items^*,\allocs^*,v,\tau)$
}
\caption{Reduction of excessive chores}\label{alg-reduction}	
\end{algorithm}


We use $Reduction(\items,\allocs,v, \tau, c^*)$ to denote the output of the Reduction algorithm. Then we prove the following lemma.

\begin{lemma}\label{lem-reduction}
Suppose that $(\items,\allocs, v,\tau)$ is First Fit Valid and there is a chore $c^*$ such that $c^*\notin \bigcup_k\alloci[k]$. Let $n=|\allocs|$. Suppose that MMS$(\items,v,n)=\mu$ and
$\tau\geq \frac{8}{7} \mu$.
Let $(\items',\allocs',v,\tau)=Reduction(\items,\allocs, v,\tau, c^*)$ and  $n'=|\allocs'|$. The tuple $(\items',\allocs',v,\tau)$ has the following properties:
\begin{enumerate}
	\item \label{proper:first-fit-valid} It is First Fit Valid.
    \item\label{proper:numberofbundle} The number of bundles $n'\le n$.
    \item\label{proper:maximin} The maximin share MMS$(\items',v,n')\le \mu$.
    \item\label{proper:unallocated} For chore $c^*$, we have $c^*\in \items'$ and $c^*\notin \bigcup_k\alloci[k]'$.
    \item\label{proper:noexcessive} For each $k$, we have $v(\alloci[k]')\le \tau$.
\end{enumerate}

\end{lemma}

\begin{proof}
~
\paragraph{1. First Fit Valid.}
By Proposition \ref{prop-tidyup-FFV}, the Tidy-Up Procedure outputs a First Fit Valid tuple when the input is First Fit Valid,
so  $(\items^*,\allocs^*,v, \tau)$
is First Fit Valid.
It remains to prove that $\mathbf{A'} = (D_1,\dots,D_k,\alloci[k+1]^*,\dots,\alloci[n']^*)$ satisfies the First Fit Valid definition for all $j\in[n']$. 
\begin{itemize}
\item For $j\leq k$, the $j$-th benchmark bundle of $\mathbf{A'}$ is the same as of $\mathbf{D}$.
 As the bundles $D_1,\dots,D_k$ are output by the FFD algorithm, all these bundles are lexicographically equal to the corresponding benchmark bundles of $\mathbf{D}$ by Proposition \ref{prop-lexiffd}.
\item For $j>k$, 
Since $c_k^{\circ}$ is a suitable reduced chore, we have $\items'\setminus\left(\bigcup_{j\le k}D_k\right)=\items^*\setminus\left(\bigcup_{j\le k}\alloci[j]^*\right)$
by Definition \ref{def-suitable-redu}. So the
$j$-th benchmark bundle of $\mathbf{A'}$ is the same as of $\mathbf{A}^*$, which is First Fit Valid.

\end{itemize} 
We can conclude that $(\items',\allocs',v, \tau)$ is First Fit Valid.

\paragraph{2. Number of bundles.}
The only step that could change the number of bundles is the Tidy-Up Procedure. This procedure never increases the number of bundles. So we have $n'\le n$.

\paragraph{3. Maximin Share.}
By Condition \ref{cond:maxvalue} of the Tidy-Up Lemma, the maximin share does not increase after the Tidy-Up Procedure. By the definition of reduced chore, we have $v(c_k^{\circ})\le v(c_k^*)$. So when we replace chore $c_k^*$ with chore $c_k^{\circ}$, the maximin share does not increase either, so MMS$(\items',v,n')\le \mu$. 

\paragraph{4. Unallocated chore.}
The chore $c^*$ is unallocated in $\mathbf{A}$.
It is never allocated by the Tidy-Up procedure, so it is unallocated in $\mathbf{A}^*$ too.
Since all reduced chores are suitable, we have $\cup_{i\in N} D_i = \cup_{i\in N} A^*_i$, 
so $c^*$ is unallocated by FFD too.
So $c^*\in \items'$ and $c^*\notin \bigcup_k\alloci[k]'$ must hold. 

\paragraph{5. No excessive chores.}
After at most $n$ iterations of the Reduction algorithm, there is no excessive chore in $\allocs'$. Therefore, $v(\alloci[k]')\le \tau$ for all $k\in[n]$.
\end{proof}

Finally, we prove our first main theorem  (\Cref{thm-approximation-intro} from the introduction):
\begin{theorem}\label{thm-tightapprox}
For any $n\geq 2$, the worst-case approximation ratio of HFFD for $n$ agents 
is the same as 
the worst-case approximation ratio 
of FFD for $n$ bins,
as shown in Table \ref{table:approxratio}.
\end{theorem}
\begin{proof}
Let $\alpha^{(n)}$ be the approximation ratio in Table \ref{table:approxratio} when the number of agents is $n$.

Suppose for contradiction that there is an instance such that $HFFD(\items,\Vec{v},\Vec{\tau})$ cannot allocate all chores, where 
the MMS of all agents is $\mu$ and $h_i=\alpha^{(n)}\cdot \mu$ for all $i \in [n]$. 
Then, by Proposition \ref{prop-HFFD-first-fit}, we have a First Fit Valid tuple $(\items,\allocs,v, \tau)$,
where $v := v_{\last} = $ the valuation of the last agent who receives a bundle in the HFFD run,
and $\tau := h_{\last} = \alpha^{(n)}\cdot \mu$,
and there is an unallocated chore $c^*\notin\bigcup_{k\le n}\alloci[k]$. As $\alpha^{(n)}\ge\frac{8}{7}$
for all $n$ in Table \ref{table:approxratio}, by Lemma \ref{lem-reduction} we can construct a First Fit Valid tuple $(\items',\allocs',v, \tau)$ that satisfies all properties 1--5 in the lemma statement.

Let $n'=|\allocs'|$ be the number of bundles. By Property \ref{proper:noexcessive}, the cost of every bundle in $\mathbf{A}'$ is not larger than $\tau = \alpha^{(n)} \mu$. By  Corollary \ref{cor-ffdsame}, these bundles are the same as the first $n'$ bundle output by $FFD(\items',v, \tau)$. 
By Property \ref{proper:maximin}, the maximin share is not larger than $\mu$, which means the threshold in $FFD(\items',v, \tau)$ is at least $\alpha^{(n)}$ times the maximin share of $v$. By Property \ref{proper:unallocated}, $FFD(\items',v, \tau)$ cannot allocate all chores. 
By Property \ref{proper:numberofbundle}, $n'\leq n$, 
so
from Table \ref{table:approxratio}, we have $\alpha^{(n')}\le\alpha^{(n)}$. So we get an instance in which  FFD cannot attain a $\alpha^{(n')}$ approximation of the minimum makespan for $n'$ bins. 
This contradicts what was already proved for job scheduling. 

In conclusion, the relationship between the number of agents and the approximation ratio in the worst case for HFFD algorithm must be the same as in Table \ref{table:approxratio}.
\end{proof}

\section{A FPTAS for nearly optimal approximation}
\label{sec:fptas}
In this section, we will demonstrate a FPTAS for the approximation of $(1+\epsilon)\frac{13}{11}$, when all costs are integers. The previous work \cite{DBLP:conf/sigecom/HuangL21} has the following observations: (1) It is very inefficient to use the PTAS for the scheduling problem to compute an approximate maximin share for each agent, and then scale them as the input for HFFD; 
(2) Due to the non-monotonicity of the MultiFit algorithm, we cannot directly use binary search (by identical costs for HFFD) to compute a ratio for each agent as the input for HFFD. 
We will show that, by a slight modification, the binary search as used by the MultiFit algorithm would yield an FPTAS for the $\frac{13}{11}\cdot \text{MMS}$.

\begin{algorithm}[H]
\KwIn{An instance $\inst$, an agent $i$, lower bound integer $l$, upper bound integer $u$}
\KwOut{A threshold for agent $i$}
\BlankLine
    \While{$l<u$}
    {
        Let $\tau :=\lfloor (l+u)/2\rfloor$\;
        Run $FFD(n, v_i, \tau)$\;
        \lIf{FFD succeeded in allocating all chores to the $n$ agents}
        {
            $u:=\tau$
        }
        \lElse{$l:=\tau+1$}
    }
	\Return {$u$}
\caption{Binary Search for One Agent's Threshold}\label{alg-binarysearch}	
\end{algorithm}

\begin{lemma}
\label{lem:binarysearch}
For any $i\in\agents$, let $h_i:=$Binary-search$(i, l_i, u_i=v_i(\items))$.

Then $MMS_i\le h_i\le \max(l_i, \alpha^{(n)}\cdot MMS_i)$.
\end{lemma}
\begin{proof}
Binary-search returns a number $h_i$ for which there exists a partition of the chores into $n$ subsets with total cost at most $h_i$ (the allocation returned by FFD). By MMS definition, this implies $MMS_i\le h_i$.

By the proof of the MultiFit approximation ratio \cite{yue1990exact},
FFD with threshold at least $\alpha^{(n)}\cdot MMS_i$ always succeeds in packing all chores into $n$ subsets.
Therefore, as long as $u> \alpha^{(n)}\cdot MMS_i$ and $u>l_i$, the algorithm will continue to another iteration;
the algorithm will stop only when 
either $u\leq  \alpha^{(n)}\cdot MMS_i$ or $u\leq l_i$.
Therefore, $h_i\le \max(l_i, \alpha^{(n)}\cdot MMS_i)$.
\end{proof}

Now, we combine the binary search with the HFFD Algorithm \ref{alg-hffd}.

\begin{algorithm}[H]
\KwIn{An IDO instance with integer costs $\inst$, an approximation parameter  $\epsilon>0$}
\KwOut{An allocation $\allocs$}
\BlankLine
    Let $R=\items$  and $\allocs=(\emptyset,\dots,\emptyset)$\;
     For all $i\in\agents$, let $h_i:=$Binary search$(i, l_i=1,u_i=v_i(\items))$\label{step:binarysearch}\;
    \While{$R\neq \emptyset$}
    {
        Run $\allocs=$HFFD$(v_1,\ldots,v_n, h_1,\ldots,h_n)$ \;
        $R=\items\setminus\allocs=$ the chores that remained  unallocated\;
        \If{$R\neq\emptyset$}
        {
            Let $\last$ be the index of the agent who is the last one getting a bundle when we run HFFD in this round of while-loop\; 
            \tcp{Notice that here $\last$ could be different in different rounds}
            Change $h_{\last}:=$Binary search$(\last, l_i = \lceil(1+\epsilon)h_{\last}\rceil, u_i = v_i(\items)$)\;\label{algline-update}
        }
    }
	\Return {$\allocs$}
\caption{A FPTAS for Approximate Maximin Share}\label{alg-ptas}	
\end{algorithm}

\begin{theorem}
\label{thm-fptas}
With integer costs, Algorithm \ref{alg-ptas} runs in time $O(\frac{1}{\epsilon}n^2m\log(V_{\text{max}}))$ and output a $(1+\epsilon)\cdot \alpha^{(n)}$-approximate MMS allocation, where $\alpha^{(n)}$ is a constant given in Table \ref{table:approxratio}.
\end{theorem}
\begin{proof}
~
\paragraph{1. Run-time analysis.}
Before the while-loop, each threshold $h_i$ is at least $MMS_i$ by Lemma \ref{lem:binarysearch}.
Once the threshold of agent $i$ becomes larger than $\alpha^{(n)}\cdot MMS_i$, this agent will never be the last one if there are remaining chores (by \Cref{thm-tightapprox}),
so $h_i$ will not be updated again.
So each threshold $h_i$ can be updated at most $\log_{1+\epsilon}\alpha^{(n)}$ times, which is $O(\frac{1}{\epsilon})$ times since $\alpha^{(n)}$ is a constant.
This means the while-loop must stop after  $O(\frac{n}{\epsilon})$ rounds. 

The binary search runs the HFFD algorithm $O(\log V_{\text{max}})$ times. The running time of HFFD algorithm can be bounded by $O(n m)$. So the total running time is $O(\frac{1}{\epsilon}n^2m\log(V_{\text{max}}))$.

\paragraph{2. Approximation ratio.}
We  prove that the threshold $h_i\le (1+\epsilon)\cdot \alpha^{(n)}\cdot MMS_i$ holds for each agent $i$ throughout the algorithm.
Initially, $h_i\le \alpha^{(n)}\cdot MMS_i$ by Lemma \ref{lem:binarysearch}.

In the while-loop, if $R\neq\emptyset$ i.e. there are some unallocated chores, then we know that $h_{\last}<\alpha^{(n)}\cdot MMS_{\last}$ holds for agent $\last$ by \Cref{thm-tightapprox}. 
The new lower bound to the binary search is $(1+\epsilon)h_{\last}< (1+\epsilon)\cdot \alpha^{(n)}\cdot MMS_{\last}$.
By Lemma \ref{lem:binarysearch},
the return value of the binary search would not be greater than $(1+\epsilon)\alpha^{(n)}\cdot MMS_{\last}$. This means that, for each agent $i$, the threshold will never be greater than $ (1+\epsilon)\cdot \alpha^{(n)}\cdot MMS_i$. Therefore, the output of Algorithm \ref{alg-ptas} is a $(1+\epsilon)\cdot \alpha^{(n)}$ approximation to the maximin share allocation.
\end{proof}

\section{Monotonicity in special cases}
\label{sec:monotonicity}
When the threshold to the FFD algorithm increases, it is natural to expect the allocation process will become easier. However, FFD might behave counter-intuitively in this respect: when it gets a larger threshold, it might allocate fewer chores. Examples were given already in \cite{DBLP:journals/siamcomp/CoffmanGJ78}.

In this section, we discuss this interesting phenomena, which we call \emph{non-monotonicity}. 
As we prove below, 
non-monotonicity is what blocks us from designing a polynomial-time algorithm for exactly $\frac{13}{11}$ approximation.
We also prove monotonicity for some special cases.

The monotonicity of FFD can be defined based on a single parameter --- the bin-size. We call this property ``weak monotonicity'', since immediately afterwards we define a stronger property.
\begin{definition}
Let $\alpha\geq 1$ be a real number and $n\geq 2$ be the number of bundles.
We say that FFD is \emph{weakly-monotone with respect to $\alpha$ and $n$} if the following holds:
\begin{quote}
For any cost function $v$ and chore set $\items$
with MMS$(\items,v,n)=\mu$, 
if FFD$(\items,v, \alpha \mu, n)$ allocates all chores, then for all $\beta\geq \alpha$, 
FFD$(\items,v,\beta \mu,n)$ allocates all chores.
\end{quote}
\end{definition}
When we consider the HFFD algorithm, the situation is more complicated since there are $n$ different thresholds.
To handle these situations, we use the notions developed in Section \ref{sec:FFabstraction} to define a stronger monotonicity property.

\begin{definition}\label{def-monotone}
Let $\alpha\geq 1$ be a real number and $n$ be the number of bundles. We say that \emph{FFD is monotone with respect to $\alpha$ and $n$} if the following holds:
\begin{quote}
For any cost function $v$ and chore set $\items$
with MMS$(\items,v,n)=\mu$, 
if FFD$(\items,v, \alpha \mu, n)$ allocates all chores, then for all First-Fit-Valid tuples $(\items,\allocs,v, \alpha \mu)$ such that $|\allocs|=n$, we have $\bigcup_{i}\alloci=\items$, that is, all chores are allocated.
\end{quote}
\end{definition}

To justify the term, we prove that it is indeed stronger than ``weak monotonicity''.
\begin{lemma}
For any $\alpha\geq 1$ and $n\geq 2$, if FFD is monotone w.r.t. $\alpha$ and $n$, then 
FFD is weakly-monotone w.r.t. $\alpha$ and $n$.
\end{lemma}
\begin{proof}
Fix a chores set $\items$ and a cost function $v$ and let $\mu := \text{MMS}(\items,v,n)$.
Suppose that FFD allocates all chores into $n$ bins of size at most $\alpha \mu$, and fix some $\beta\geq \alpha$.
Let $\mathbf{D}:= \text{FFD}(\items, v, \beta \mu)$.
Then the tuple $(\items, \mathbf{D}, v,\beta \mu)$ is First-Fit-Valid. This means that, for all $j\in[n]$, $D_j$ is (weakly) lexicographically larger than all subsets in $\items\setminus\cup_{i<j}A_i$ with a cost of at most $\beta \mu$. 
Since $\alpha\leq \beta$, 
$D_j$ is (weakly) lexicographically larger than all subsets in $\items\setminus\cup_{i<j}A_i$ with a cost of at most $\alpha \mu$. 
Therefore, the tuple $(\items, \mathbf{D}, v, \alpha \mu)$ is First-Fit-Valid too.
By definition of monotonicity, this implies $\bigcup_{i}D_i=\items$. This means that FFD allocates all chores into the $n$ bundles in $\mathbf{D}$, which are of size $\beta \mu$. So FFD is weakly-monotone w.r.t. $\alpha$ and $n$.
\end{proof}

The following lemma explains why monotonicity is useful from an algorithmic perspective.
\begin{lemma}
\label{lem-monotonicity-implies-exact-mms}
For any $n\geq 2$ and constant $r < 13/11$, if FFD is monotone w.r.t. $n$ and all $\alpha\geq r$, 
then \Cref{alg-ptas} 
with  any $\epsilon \leq (\frac{13}{11}/r)-1$
outputs a $\frac{13}{11}$ MMS allocation.
\end{lemma}

\begin{proof}
Similarly to the proof of Theorem \ref{thm-fptas}, 
it is sufficient to prove that $h_i\le \frac{13}{11}MMS_i$ holds for each agent $i$ throughout the algorithm.
The binary search in line \ref{step:binarysearch} gives an initial $h_i$ for each agent $i$ such that $MMS_i\le h_i\le
 \alpha^{(n)}MMS_i
 \le \frac{13}{11}MMS_i$
 by 
Lemma \ref{lem:binarysearch}.

For every agent $i$, the threshold $h_i$ is computed such that FFD$(\items, v_i, h_i)$ allocates all chores.
By definition of monotonicity, if $h_i\geq r\cdot MMS_i$, then in every First Fit Valid tuple $(\items, \allocs, v_i, h_i)$, all chores are allocated.
Conversely, if HFFD in the while-loop does \emph{not} allocate all chores, then 
$h_{\last}<r\cdot MMS_{\last}$ must hold, 
since  $(\items, \allocs, v_{\last}, h_{\last})$ is First Fit Valid by \Cref{prop-HFFD-first-fit}.

The new lower bound to the binary search is $
(1+\epsilon)\cdot h_{\last}
<
(1+\epsilon)\cdot r \cdot MMS_{\last}
\leq
\frac{13}{11} MMS_{\last}
$, by the condition on $\epsilon$.
During the binary search, FFD always succeeds when 
$h_{\last}\geq \frac{13}{11}MMS_{\last}$.
So the return value of the binary search would be at most $\frac{13}{11}MMS_{\last}$. 
\end{proof}

For every fixed constant $r<\frac{13}{11}$, 
the 
$\epsilon := (\frac{13}{11}/r)-1$ is constant too, so the run-time 
of \Cref{alg-ptas} is 
in $O(n^2 m \log(V_{\text{max}}))$.
So monotonicity would imply a polynomial-time algorithm for $\frac{13}{11}$-MMS chore allocation.

\Cref{exm:chores} shows that FFD is not monotone for $n=4$ and $\alpha=1$.
Does the proof of Lemma \ref{lem-reduction} imply that FFD is monotone for all $\alpha \geq \frac{8}{7}$? 
The answer is no. 
Recall that our proof is by contradiction:
given 
an instance in which HFFD
with unallocated chores ($\bigcup_{i}\alloci\subsetneq\items$)
we can construct a new instance for which FFD with the same threshold leaves some chores unallocated. 
This argument does not work for proving monotonicity, 
because we do not know anything about the behavior of FFD in the original instance,
so we can say nothing about monotonicity. 


Though we still have technical difficulties in proving monotonicity, the threshold in all counter-examples we could find is very close to $1$. Therefore, we make the following conjecture. 

\begin{conjecture}\label{conj:montone}
FFD is monotone with respect to any $n\geq 2$ and any $\alpha\ge\frac{8}{7}$.
\end{conjecture}
Conjecture \ref{conj:montone} combined with \Cref{lem-monotonicity-implies-exact-mms}
would imply that running \Cref{alg-ptas}
with $\epsilon \approx 0.034$ is sufficient for attaining a $\frac{13}{11}$-MMS approximation.


To shed the light on Conjecture \ref{conj:montone}, we prove the following theorem. 
\begin{theorem}
\label{thm-ffd-is-monotone}
FFD is monotone in the following cases:

(a) With respect to $n=2$ and any $\alpha\geq 1$.

(b) With respect to $n=3$ and any $\alpha\geq \frac{10}{9}$. 
\end{theorem}

By arguments analogous to the proof of \Cref{lem-monotonicity-implies-exact-mms},
the theorem implies that, if we change $\frac{13}{11}$ to $\alpha^{(n)}$ in 
\Cref{alg-ptas}, then \Cref{alg-ptas} can find an exact $\alpha^{(n)}$-approximation in polynomial time for $n=2,3$ if all costs are integers, which proves \Cref{thm-monotonicity-intro} from the introduction.

We prove \Cref{thm-ffd-is-monotone} in Appendix \ref{app:benchmark}.


\section{Future Work}

By reducing the chore allocation problem to the job scheduling problem, we give a tight analysis of HFFD algorithm. This opens the door to look at the chore allocation problem from the angle of the job scheduling problem. An interesting direction for the future work is to explore more connections between these two problems. As the job scheduling problem is well studied, it is possible we can find other algorithm of the job scheduling algorithm that we can apply to chores  allocation problems. 

We leave an unsolved problem about the monotonicity of FFD for $n\geq 4$ to the reader. A positive answer to this question could improve the performance of \Cref{alg-ptas}.
As a step towards that goal, it may be useful to develop a concept a minimal counter-example, or a ``tidy-up tuple", for monotonicity.

Interestingly, we do not even know if the monotonicity property itself is monotone: 
suppose that FFD is monotone with respect to $\alpha$ and $n$ --- does this imply that FFD is monotone with respect to $\beta$ and $n$ for any $\beta>\alpha$?

\section*{Acknowledgements}
Xin Huang is supported in part at the Technion by an Aly Kaufman Fellowship.
Erel Segal-Halevi is supported by the Israel Science Foundation (grant no. 712/20).
We are grateful to the reviewers in EC 2023 for their helpful comments.


\bibliographystyle{alpha}
\bibliography{sample}

\newpage

\appendix
\section*{APPENDIX}

\section{Finding a suitable reduced chore}
\label{app:reduced-chore}

In this section we prove 
Lemma \ref{lem-suitable-reduced-chore}
from Subsection  \ref{sub:reduction}.
Let $(\items, \allocs, v,\tau)$ be a Tidy-Up tuple with $\mu := \text{MMS}(\items,v,n)$
and $\tau \ge \frac{8}{7}\mu$,
and at least one bundle has an excessive chore. We show that we can construct a suitable reduced chore (\Cref{def-suitable-redu}). 
Throughout the proof we denote $\gamma := \tau-\mu$, so $\gamma \geq \frac{1}{7}\mu$.

To simplify the reduction, we introduce a value called \emph{fit-in space}. 
 It is  the largest cost that a reduced chore can have such that it is not an excessive chore. 

\begin{definition}[Fit-in space]
Given an index $k$ and a threshold $\tau$, let $c'_k$ be the last chore in bundle $\alloci[k]$.  The \emph{fit-in space} of bundle $\alloci[k]$ is defined as $fs(k)=\tau-v(\alloci[k]\setminus c'_k)$.
\end{definition}
Note that, in a Tidy-Up tuple with
MMS $=\mu\le \tau$,
the fit-in space is always at least 0, since there are no $\mu$-redundant
chores, so removing any single chore brings the bundle cost below $\tau$.

\begin{lemma}
\label{lem:less-than-fs}
Let $(\items, \allocs, v, \tau)$ be a First-Fit-Valid tuple. 
Let $i\in[k]$ and $c_i$ any chore in bundle $A_i$.
Let $j>i$ and $c_j$ any chore in bundle $A_j$.
If $v(c_j)\leq fs(i)$,
then $v(c_j)\leq v(c_i)$.
\end{lemma}
\begin{proof}
Define $A_i' := A_i \setminus \{c_i\} \cup \{c_j\}$.
The condition  $v(c_j)\leq fs(i)$ implies that $v(A_i')\leq \tau$.

By First-Fit-Valid definition, this implies $A_i\lex{\ge} A_i'$.
Since the only difference between $A_i$ and $A_i'$ is the replacement of $c_i$ by $c_j$,  This implies $v(c_i)\geq v(c_j)$.
\end{proof}

We prove now that we can find a suitable reduced chore. We consider three cases, depending on the size of the fit-in space $fs(k)$.

\subsection{Small fit-in space}
We first consider the case in which the fit-in space is smaller than the non-allocated chore.

\begin{lemma}
\label{lem-suitable-with-small-fs}
Let $(\items, \allocs, v, \tau)$ be a 
First Fit Valid tuple
with an unallocated chore $c^*$,
where there are no $\tau$-redundant chores 
and no chores smaller than $c^*$.
Let $k$ be the smallest integer such that there is an excessive chore $c_k^*\in\alloci[k]$. 

If  $fs(k) < v(c^*)$, then a chore $c_k^{\circ}$ with cost $0$ is a suitable reduced chore for $c_k^*$.
\end{lemma}
\begin{proof}

Since the cost of $c_k^{\circ}$ is $0$, it is placed in bundle $D_1$, so we have $D_1=\alloci[1]\cup\{c_k^{\circ}\}$.
By \Cref{cor-ffdsame}, bundles $\alloci[1],\dots, \alloci[k-1]$ are the same as the output of FFD algorithm on $\alloci[1]\cup \cdots \cup \alloci[k-1]$. So for the bundles $D_1,\dots, D_{k-1}$, we have
$\bigcup_{j<k}D_j=\bigcup_{j<k}\alloci[j]\cup\{c_k^{\circ}\}$.
When FFD constructs bundle $D_k$,
After allocating chores $\alloci[k]\setminus\{c_k^*\}$, 
the remaining space  is $fs(k)$. 
However, as $fs(k)<v(c^*)$, 
and $c^*$ is the smallest chore in $\items$ by the lemma assumption,
no remaining chore can be allocated to $D_k$. This implies $D_k=\alloci[k]\setminus\{c_k^*\}$, so Definition \ref{def-suitable-redu} holds.
\end{proof}

\subsection{Medium fit-in space}
Next, we consider the case in which the fit-in space is larger than the unallocated chore, but smaller than $2 \gamma$
(recall that $\gamma := \tau-\mu$).

\begin{lemma}
\label{lem-suitable-with-medium-fs}
Let $(\items, \allocs, v, \tau)$ be a First-Fit-Valid tuple
with an unallocated chore $c^*$,
where there are no $\tau$-redundant chores and no chores smaller than $c^*$.
Let $k$ be the smallest integer such that there is an excessive chore $c_k^*\in\alloci[k]$. 
If  $v(c^*)\le fs(k)\le 2\gamma$, then we can construct for $c_k^*$ a suitable reduced chore $c_k^{\circ}$.
\end{lemma}
\begin{proof}
To define the cost of the reduced chore $c_k^{\circ}$, we introduce two disjoint sets. The first set contains all non-last chores allocated before the excessive chore $c_k^*$. Formally:
\begin{align*}
NT:=\{c'\in\bigcup_{j\le k}\alloci[j]\mid c'\text{ is a non-last chore}\}
\end{align*}
The second set contains the chores that have not been allocated yet but they can compete for the fit-in space in bundle $\alloci[k]$. Formally:
\begin{align*}
FT:=\{c\in\items\setminus(\bigcup_{j\le k}\alloci[j])\mid v(c)\le fs(k)\}
\end{align*}

Before we define the cost of reduced chore, we informally present the two properties that it should satisfy. Then we give a precise definition and prove  these two properties hold.
\begin{itemize}
\item First, we need $v(c_k^{\circ})\le \min_{c'\in NT}\{v(c')\}$, so that the allocation of $\items'$ would not influence any non-last chore allocated before chore $c_k^*$.
\item Second, we need $v(c_k^{\circ})\ge\max_{c\in FT}\{v(c)\}$, so that  any chore in the set $FT$ would not influence the first $k$ bundles of the allocation of  $\items'$.
\end{itemize} 

Let us first prove that it is possible to satisfy both properties at once.
\begin{claim}
\label{claim-cklowerbound}
$\min_{c'\in NT}\{v(c')\} \geq \max_{c\in FT}\{v(c)\}$.
\end{claim}

\begin{proof}
We prove that $v(c')\geq v(c)$ for every $c'\in NT$ and $c\in FT$.
We consider two cases.

If $c'$ is a non-last chore of bundle $\alloci[k]$,
then $v(c')\geq v(c^*_k)$ since $c^*_k$ is the last chore of 
$\alloci[k]$,
and $v(c^*_k)>fs(k)$ since $c^*_k$ is an excessive chore,
and $fs(k)\geq v(c)$ by definition of $FT$. By transitivity, $v(c')\geq v(c)$.

Otherwise, $c'$ is a non-last chore of a bundle $\alloci[j]$, where $j< k$.
Since $A_k$ is the first bundle with an excessive chore, $A_j$ has no excessive chore, so $v(A_j)\leq \tau$. We now construct from $A_j$ a new bundle $A_j'$ as follows:
\begin{itemize}
\item Remove from $A_j$ the chore $c'$ and the last chore of $A_j$.  
Each chore cost is larger than 
$v(c^*)$ by the assumption of 
Lemma \ref{lem-suitable-with-medium-fs},
which is larger than $\gamma$
by Lemma \ref{lem-nosmall}.
Therefore, the cost after the removal is smaller than $\tau-2\gamma$.
\item Add the chore $c$. Since $c\in FT$, $v(c)\leq fs(k)$. By assumption, $fs(k)\leq 2\gamma$. Therefore, the cost after the addition is at most $\tau$.
\end{itemize} 
Since the tuple is First Fit Valid, $A_j$ is lexicographically-maximal among all bundles with cost at most $\tau$ in $\items\setminus \cup_{i<j}A_i$. In particular, 
$A_j \lex{\geq} A_j'$. 
Note that $A_j$ and $A_j'$ are identical, except that $c'$ and a smaller chore are replaced with $c$. Hence, $A_j \lex{\geq} A_j'$ requires that $v(c')\geq v(c)$.
\end{proof}

We now define the cost of the reduced chore as
\begin{align*}
v(c_k^{\circ}):=\min\{fs(k),\min_{c\in NT}\{v(c)\} \}.
\end{align*} 
It satisfies the first property by definition, and satisfies the second property by Claim \ref{claim-cklowerbound} and the definition of FT.
We prove below that it is a suitable reduced chore.

Let $\items'=\items\setminus\{c^*_k\}\cup\{c_k^{\circ}\}$. 
Let $\mathbf{D}$ be the output of $FFD(\items',v, \tau)$. 
We now prove that the first $k$ bundles of $\mathbf{D}$ are almost the same as the first $k$ bundles of $\mathbf{A}$: they can differ only in the last chores. Formally:
\begin{claim}
\label{claim-non-last-chores}
For any $j\in\{1,\ldots, k\}$, all non-last chores of $A_j$ are allocated to $D_j$.

\end{claim}
\begin{proof}
By Corollary \ref{cor-ffdsame}, 
since $(\items,\allocs,v, \tau)$ is First Fit Valid, the bundles $(\alloci[1],\dots,\alloci[k-1])$ are the first $k-1$ bundles of the output of $FFD(\items,v, \tau)$. 
Moreover, $FFD(\items,v, \tau)$ allocates all non-last chores in $A_k$ to the $k$-th bundle.
Therefore, we can compare the non-last chores in $A_1,\ldots,A_k$ with those in $D_1,\ldots,D_k$ by comparing the two runs of FFD. These two runs differ only due to the removal of $c_k^*$ and the addition of $c_k^{\circ}$.

The removal of $c_k^*$ does not affect bundles $1,\ldots, k-1$ at all, since $c_k^*$ was not allocated to them by $FFD(\items,v, \tau)$. 
It does not affect the non-last chores in $A_k$, since all of them were processed before $c_k^*$ by $FFD(\items,v, \tau)$.

The addition of $c_k^{\circ}$ does not affect the non-last chores in $A_1,\ldots, A_{k-1}$, since these chores are in NT, and $c_k^{\circ}$ is smaller than them by construction, so FFD processes $c_k^{\circ}$ after all these chores have already been allocated.
The non-last chores in $A_k$ are larger than $c_k^{\circ}$, so they too are processed before $c_k^{\circ}$.
\end{proof}

\begin{figure}
    \centering
    \includegraphics[width=16cm]{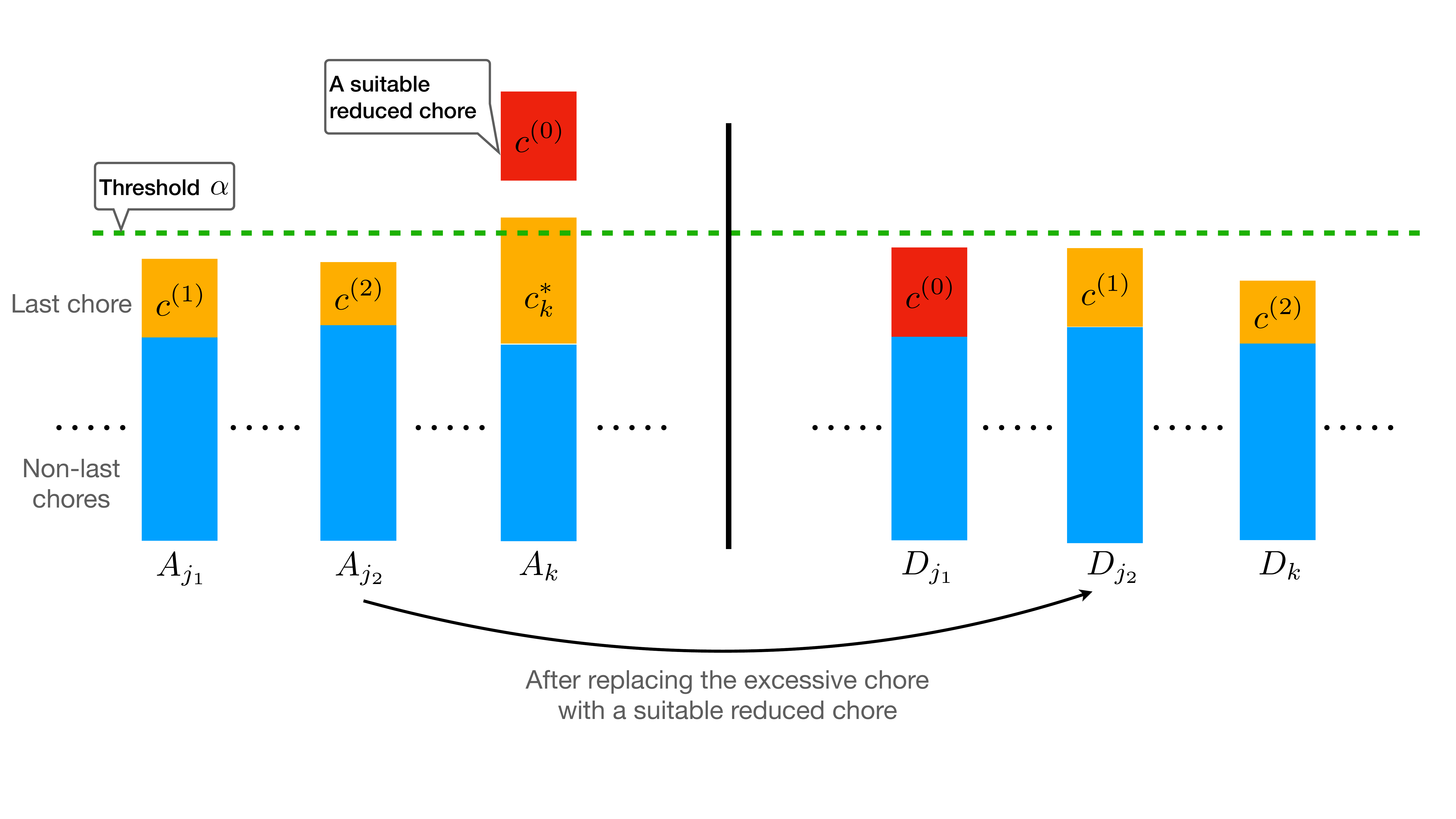}
    \caption{The sequence of reallocation after replacing the excessive chore $c_k^*$  with a suitable reduced chore $c^{\circ}$, in the proof of Claim \ref{claim:fs(k)>v(c*)} (for a medium fit-in space). Here we have $t=2$.}
    \label{fig:sequence}
\end{figure}
\begin{claim}
\label{claim:fs(k)>v(c*)}
When $v(c^*)\leq fs(k) \leq 2\gamma$, the equation 
in Definition \ref{def-suitable-redu} holds true, so chore $c_k^{\circ}$ is a suitable reduced chore.
\end{claim}
\begin{proof}
By Claim \ref{claim-non-last-chores}, the non-last chores in bundles $A_1,\ldots,A_k$ are allocated to the same bundles in $D_1,\ldots,D_k$.
We prove that the only possible effect of the reduction is a cyclic exchange of last chores
among bundles $1,\ldots,k$. 
Please refer to Figure \ref{fig:sequence} for an illustration of such a cyclic exchange, involving three  chores.



Denote $c^{(0)}:=c_k^{\circ}$.
While $c^{(0)}$ is not allocated, 
the bundles $D_1,D_2,\ldots$ are equal to the bundles $A_1,A_2,\ldots$. 
If $c^{(0)}$ is not allocated before bundle $D_k$, then there is room for it there since $v(c^{(0)})\leq fs(k)$.
Moreover, by Claim \ref{claim-cklowerbound} we have 
$v(c^{(0)})\geq \max_{c\in FT}\{v(c)\}$, so 
$v(c^{(0)})$ is processed before any chore that could compete with it on $fs(k)$.
Therefore, FFD allocates $c^{(0)}$ into $D_k$.

Otherwise, $c^{(0)}$ is  allocated to $D_{j_1}$ for some $j_1\in\{1,\ldots,k-1\}$.
We denote by $c^{(1)}$ the last chore in bundle $A_{j_1}$.
As non-last chores are not affected by the reduction, only two cases are possible: 1. $D_{j_1}=\alloci[j_1]\cup\{c^{(0)}\}$; 2. $D_{j_1}=\alloci[j_1]\setminus\{c^{(1)}\}\cup \{c^{(0)}\}$. 
In fact, as $v(c^{(0)})\ge v(c^*)$ and the tuple $(\items,\allocs,v, \tau)$ is First Fit Valid, we have $v(\alloci[j_1]\cup\{c^{(0)}\})
\geq
v(\alloci[j_1]\cup\{c^*\})
>\tau$, so 
case 1 is not possible --- the only possible case is that  $D_{j_1}=\alloci[j_1]\setminus\{c^{(1)}\}\cup \{c^{(0)}\}$, that is, chore $c^{(0)}$ ``pushes'' chore $c^{(1)}$ out of its bundle. 
This implies that FFD processes $c^{(0)}$ before $c^{(1)}$,
so $v(c^{(0)})\geq v(c^{(1)})$.

Now, we apply similar reasoning to $c^{(1)}$.
While $c^{(1)}$ is not allocated, 
the bundles $D_{j_1+1},D_{j_1+2},\ldots$ are equal to the bundles $A_{j_1+1},A_{j_1+2},\ldots$. 
If $c^{(1)}$ is not allocated before  $D_k$, then there is room for it there since $v(c^{(0)})\geq v(c^{(1)})$.
Moreover,
since $c^{(0)}$ replaced $c^{(1)}$ in bundle $D_{j_1}$, 
We know that 
the fit-in space of bundle $A_{j_1}$ must satisfy $fs(j_1)\geq v(c^{(0)})\geq \max_{c\in FT} v(c)$,
so $v(c)\leq fs(j_1)$ for any chore  $c\in FT$.
\Cref{lem:less-than-fs} implies that $v(c)\leq v(c^{(1)})$
for any $c\in FT$,
so $v(c^{(1)})\geq \max_{c\in FT}\{v(c)\}$.
This also implies that, when FFD constructs $D_k$,
$c^{(1)}$ is processed before any chore that could compete with it on $fs(k)$.
Therefore, FFD allocates $c^{(1)}$ into $D_k$, and we are done  --- we have completed an exchange cycle of length 2.

Otherwise, $c^{(1)}$ is necessarily allocated to $D_{j_2}$ for some $j_2\in\{j_1+1,\ldots,k-1\}$.
We denote by $c^{(2)}$ the last chore in bundle $A_{j_2}$.
As non-last chores are not affected by the reduction, only two cases are possible: 1. $D_{j_2}=\alloci[j_2]\cup\{c^{(1)}\}$; 2. $D_{j_2}=\alloci[j_2]\setminus\{c^{(2)}\}\cup \{c^{(1)}\}$. 
In fact, as $v(c^{(1)})\ge v(c^*)$ and the tuple $(\items,\allocs,v, \tau)$ is First Fit Valid, we have $v(\alloci[j_1]\cup\{c^{(1)}\})
\geq
v(\alloci[j_1]\cup\{c^*\})
>\tau$, so 
case 1 is not possible --- the only possible case is that  $D_{j_2}=\alloci[j_2]\setminus\{c^{(2)}\}\cup \{c^{(1)}\}$, that is, chore $c^{(1)}$ ``pushes'' chore $c^{(2)}$ out of its bundle, which implies $v(c^{(1)})\geq v(c^{(2)})$.

By a similar argument, we can define $j_3<j_4<\dots< j_t<k$ and $c^{(3)},c^{(4)},\dots,c^{(t)}$ with 
$v(c^{(1)}) \geq v(c^{(2)}) \geq \cdots \dots \geq v(c^{(t)})\geq
\max_{c\in FT} v(c)\geq
 v(c^*)$. 
 Here, $t$ is the integer for which chore $c^{(t)}$ is allocated into $D_k$. Note that $t=0$ is possible.
 
In summary,
for every $i$ in $1,\ldots, t$, 
we have $D_{j_i} = A_{j_i} \setminus \{c^{(i)}\}
\cup \{c^{(i-1)}\}$,
and $c^{(t)}$ is allocated to $D_k$.
Finally, as the fit-in space $fs(k)\le 2\gamma$, and all chore costs are larger than $\gamma$, there is no space to allocate another chore into $D_k$ after allocating  $c^{(t)}$. So we have $D_k=A_k\setminus \{c_k^*\}\cup \{c^{(t)}\}$,  
 and Definition \ref{def-suitable-redu} is satisfied.
\end{proof}

This completes the proof of Lemma \ref{lem-suitable-with-medium-fs}.
\end{proof}

\subsection{Large fit-in space}
The case in which the fit-in space is large requires a different proof, since in this case the non-last chores might also be affected.

To prove this case, following  \cite{DBLP:journals/siamcomp/CoffmanGJ78},
we classify the allocated chores as follows. Given a Tidy-Up tuple $(\items,\allocs,v, \tau)$, for any  chore $c\in \alloci[k]$ for some $k\in[n]$:
\begin{itemize}
\item Chore $c$ is called an \emph{excessive chore} if it satisfies Definition \ref{def-excessive}, that is, it is the last (smallest) chore in $A_k$, and $v(A_k)>\alpha$.
~~~
In Example \ref{exm:chores}, running HFFD with threshold 75 and taking the cost function of the type A agents, chore 6 in bundle $A_1$ is  excessive.

\item Otherwise, $c$ is called a \emph{regular} chore if its cost is greater than or equal to the largest chore in the next bundle,
that is, $\vai[]{c}\ge\vai[]{\alloci[k+1][1]}$.
Note that, by First-Fit-Valid properties, this implies that a regular chore is larger than \emph{all} chores in \emph{all} next bundles.
~~~
In Example \ref{exm:chores}, the regular chores are $c_1$ (in $A_1$), $c_2,c_3$ (in $A_2$), $c_4,c_5,c_8$ (in $A_3$), and all chores in $A_4$.
\item 
 Otherwise, $c$ is called a \emph{fallback} chore. 
 In Example \ref{exm:chores}, the only fallback chore is $c_7$ in $A_2$.
\end{itemize}
For any integer $r\geq 1$, a bundle $\alloci[k]$ containing exactly $r$ regular chores
is called \emph{$r$-regular}. 
\cite{DBLP:journals/siamcomp/CoffmanGJ78} prove the following lemma for the output of FFD; we extend the proof to Tidy-Up tuples:
\begin{lemma}
\label{lem-regular-bundles}
In any Tidy-Up tuple $(\items,\allocs,v, \tau)$,
for every integer $r\geq 1$, all the $r$-regular bundles in $\mathbf{A}$ appear before all the $r+1$-regular bundles.
\end{lemma}
\begin{proof}
Suppose some bundle $A_j$ is $r$-regular. We prove that $A_{j+1}$ contains at least $r$ regular chores.

Since regular chores are not excessive, the sum of regular chores in every bundle is at most $\tau$.
Therefore, $v(A_j[r])\leq \tau/r$, where 
$A_j[r]$ is the smallest regular chore in $A_j$.
Since $A_j[r]$ is regular, its cost is at least as large as all chores in $\items\setminus \cup_{i\leq j}A_{i}$.
Let $L_r$ be the set of $r$ largest chores in 
$\items\setminus \cup_{i\leq j}A_{i}$.
Since all these chores' costs are at most $\tau/r$, we have
$v(L_r)\leq r\cdot (\tau/r) = \tau$.
By First Fit Vaild properties, 
$A_{j+1}$ is (weakly) lexicographically larger 
than all subsets of $\items\setminus \cup_{i\leq j}A_{i}$ with cost at most $\tau$.
Therefore, $A_{j+1}$ must contain $L_r$.
Note that there must be indeed at least $r$ chores in $\items\setminus \cup_{i\leq j}A_{i}$, since there is an unallocated chore.
All the chores in $L_r$ are regular in $A_{j+1}$, which completes the proof.
\end{proof}

\begin{lemma}
\label{lem-suitable-with-large-fs}
Let $(\items, \allocs, v, \tau)$ be a Tidy-Up tuple (where the MMS is $\mu$). Let $\gamma := \tau - \mu$, and assume $\gamma \geq \mu/7$.
Let $k$ be the smallest integer such that there is an excessive chore $c_k^*\in\alloci[k]$. If the fit-in space  $fs(k)> 2\gamma$,
then a chore $c_k^{\circ}$ with cost
$v(c_k^{\circ}):=fs(k)$ is a 
suitable reduced chore for $c_k^*$.
\end{lemma}

\begin{proof}
Throughout the proof, we denote by $\mathbf{P}$ the partition constructed by the Tidy-Up procedure (where it was denoted by $\mathbf{P'}$).

First, we bound the cost of the largest chore in a 1-regular bundle. 
\begin{claim}\label{claim-1st-item-value}
Suppose that bundle $\alloci[j]$ is 1-regular. Then $\frac{\mu+\gamma}{2}<v(\alloci[j][1])<\mu-2\gamma$.
\end{claim}
\begin{proof}
We first prove that $v(\alloci[j][1])>\frac{\mu+\gamma}{2}$.
By Condition \ref{cond:atleasttwochores} in the Tidy-Up Lemma,  $\alloci[j]$ contains at least two chores. 
Because bundle $\alloci[j]$ is 1-regular,  $\alloci[j][2]$ is not a regular chore. So there is a chore  $c$ not allocated in first $j$ bundles with  $v(c) > v(\alloci[j][2])$. 
So the bundle $\{ \alloci[j][1], c\}$ is leximin-larger than $\alloci[j]$.
By \Cref{lem-lex-larger-implies-cost-larger}, this implies
$v(\{ \alloci[j][1], c\})>\tau = \mu+\gamma$.
Since $\alloci[j][1]$ is a regular chore, the cost $v(\alloci[j][1])\ge v(c)$. Therefore,  $v(\alloci[j][1])>\frac{\mu+\gamma}{2}$.

Next, we prove that $v(\alloci[j][1])<\mu-2\gamma$. Consider the partition $\mathbf{P}$
constructed by the Tidy-Up procedure,
and let $P_i$ be the bundle that contains chore 
$\alloci[j][1]$. 
By Condition \ref{cond:threechores} and \ref{cond:maxvalue} in the Tidy-Up Lemma, 
we have $|P_i|\ge 3$ and $v(P_i)\le \mu$. 
By Condition \ref{cond:hz}, all chores are larger than  $\gamma$. 
So  $v(\alloci[j][1])<v(P_i)-2\gamma \le \mu-2\gamma$.
\end{proof}


In the next claim, we show that the bundle $\alloci[k]$ 
(the first bundle with an excessive chore) must be 2-regular and not contain fallback chores.
\begin{claim}\label{claim-ak-2-regular}
For bundle $\alloci[k]$, we have
$\alloci[k]=\{c_1,c_2,c_k^*\}$, where $c_1$ and $c_2$ are regular chores and $c_k^*$ is an excessive chore.
\end{claim}
\begin{proof}
First, we prove that  $|\alloci[k]|\le 3$. 

Since $c_k^*$ is excessive, its cost is larger than the fit-in space $fs(k)$, which is larger than $2\gamma$ by the assumption of \Cref{lem-suitable-with-large-fs}.
As $c_k^*$ is the smallest chore in the bundle, the cost of every other chore in $A_k$ is larger than $2\gamma$ too.
By definition of fit-in space,
$v(A_k\setminus c_k^*) + fs(k) = \tau = \mu+\gamma$,
so 
$v(A_k\setminus \{c_k^*\}) = \mu+\gamma-fs(k) < \mu+\gamma-2\gamma = \mu-\gamma$.
Since $\gamma\geq \mu/7$, 
$A_k\setminus \{c_k^*\}$ can contain at most two chores.
So we have $|\alloci[k]|\le 3$.

It is also impossible that $|\alloci[k]|\leq  2$, since by Condition \ref{cond:twochores} of the Tidy-Up Lemma, the sum of costs of the two largest chores in $\alloci[k]$ is less than $\mu$, while $v(A_k)>\tau$ since $A_k$ has an excessive chore.
Therefore, we have $|\alloci[k]|=3$. 

Since $c_k^*$ is an excessive chore, bundle $\alloci[k]$ can only be 2-regular or 1-regular.  
Suppose for contradiction that it is 1-regular.  
By Claim \ref{claim-1st-item-value}, the cost $v(\alloci[k][1])>\frac{\mu+\gamma}{2}$. 
Consider the partition $\mathbf{P}$ generated at the end of  the Tidy-Up Procedure.
Let $P_k$ be the bundle that contains the chore $\alloci[k][1]$. We first prove that $|P_k|=3$. The remaining space in bundle $P_k$ is 
at most
$\mu-v(\alloci[k][1])<\frac{\mu-\gamma}{2}$. As $\gamma\ge\frac{\mu}{7}$, we have $\frac{\mu-\gamma}{2}\le 3\gamma$. This implies that there are at most two more chores in bundle $P_k$ except chore $\alloci[k][1]$, because the cost of every chore is greater than $\gamma$. This implies $|P_k|\le 3$. By Condition \ref{cond:threechores} in the Tidy-Up Lemma $|P_k|\ge 3$, so  $|P_k|=3$. 

The cost of 
every
chore in $P_k$ except $\alloci[k][1]$ is at most $\mu-[v(\alloci[k][1])+\gamma]$, which is smaller than $2\gamma$. As $c_k^*$ is the smallest chore in bundle $A_k$, and
$v(c_k^*)>fs(k)>2\gamma$,
the cost of every
chore in $A_k$ except $\alloci[k][1]$ is larger than $2\gamma$.
Therefore,  
bundle $\alloci[k]$ dominates bundle $P_k$, 
which contradicts the operation of the Tidy-Up Procedure.

The only remaining case is that $\alloci[k]$ is 2-regular, so $c_1$ and $c_2$ are regular chores.
\end{proof}

\begin{claim}
\label{claim-smaller-3gamma}
The second-highest chore cost in $A_k$ is smaller than $3\gamma$, 
$v(\alloci[k][2])< 3\gamma$.
\end{claim}
\begin{proof}
By Claim \ref{claim-ak-2-regular} and definition of fit-in space, $v(\{\alloci[k][1], \alloci[k][2]\})=\mu+\gamma-fs(k)$.
So $v(\{\alloci[k][1], \alloci[k][2]\})<\mu-\gamma$.
So $v(\alloci[k][2])<\frac{1-\gamma}{2}$. When $\gamma\ge\frac{1}{7}$, we have $\frac{1-\gamma}{2}\le 3\gamma$. 
By transitivity, $v(\alloci[k][2])< 3\gamma$. 
\end{proof}

Next we show a relationship between the second largest chore in a 1-regular bundle and the excessive chore.
\begin{claim}
\label{claim-2nd-item-value}
Suppose that bundle $\alloci[j]$ is 1-regular. 
Then 
$v(\alloci[j][2])\ge v(c_k^*)$.
\end{claim}

\begin{proof}
Since $v(A_k[2])\geq v(c_k^*)$, it is sufficient to prove $v(\alloci[j][2])\ge v(\alloci[k][2])$. 

By Claim \ref{claim-1st-item-value},  $v(\alloci[j][1])<\mu-2\gamma$. 
By Claim \ref{claim-smaller-3gamma},
$v(\alloci[k][2])<3\gamma$.
Therefore, the bundle $A_j' := \{A_j[1], A_k[2]\}$ has cost smaller than $\mu+\gamma = \tau$.
Since $A_k$ is 2-regular, $j<k$ by \Cref{lem-regular-bundles}.
Therefore, by First-Fit-Valid properties, $A_j \lex{\geq} A_j'$.
Since $A_j[1]$ occurs in both bundles, we must have $v(A_j[2])\geq v(A_k[2])$.
\end{proof}

Recall that we define the cost of the reduced chore as $v(c_k^{\circ}):=fs(k)$. 
So $v(\alloci[k]\setminus \{c_k^*\}\cup\{c_k^{\circ}\})=\tau$. Let $\items'=\items\setminus\{c_k^*\}\cup\{c_k^{\circ}\}$ and $\mathbf{D}=FFD(\items',v, \tau)$.

\begin{claim}\label{claim-notin-1regular}
For any $j<k$ for which $A_j$ is 1-regular, FFD will not allocate $c_k^{\circ}$ to $D_j$.
\end{claim}
\begin{proof}
Denote by $t$ the number of 1-regular bundles in $\mathbf{A}$.
We prove the claim by induction on $t$.
For $t=0$ the claim holds vacuously.
Otherwise, 
By \Cref{lem-regular-bundles}, 
the 1-regular bundles in 
$\mathbf{A}$ appear first, so these are exactly the bundles $A_1,\ldots,A_t$.
Suppose the claim holds for $t-1$ 1-regular bundles. By the induction assumption, FFD does not allocate $c_k^{\circ}$ to bundles 
$D_1,\ldots,D_{t-1}$, so
they are equal to $A_1,\ldots, A_{t-1}$.
We prove the same holds for $D_t$.

By Claim \ref{claim-2nd-item-value},
The cost of the reduced chore $c_k^{\circ}$ is less than chore $\alloci[t][2]$,
so it does not influence the allocation of chore $\alloci[t][2]$ or any earlier chore.
So the first two chores in $D_t$ are $\alloci[t][1]$ and $\alloci[t][2]$. 

On the other hand, we prove that after allocating the chore $\alloci[t][2]$ to $D_t$, there is no space to allocate the reduced chore $c_k^{\circ}$. By Claim \ref{claim-1st-item-value}, $v(\alloci[t][1])>\frac{\mu+\gamma}{2}$. 
By Claim \ref{claim-2nd-item-value}, 
$v(\alloci[t][2])\ge v(c_k^*)>v(c_k^{\circ})=fs(k)>2\gamma$. 
So $v(\{\alloci[t][1], \alloci[t][2]\})>0.5 \mu+2.5 \gamma$. The remaining space in $D_t$ after allocating first two chores is $\mu+\gamma-v(\{\alloci[t][1], \alloci[t][2]\})<0.5 \mu-1.5\gamma$. As $\gamma\ge\frac{\mu}{7}$, we have $3.5\gamma \geq 0.5 \mu$, so $0.5z-1.5\gamma\le 2\gamma$. Since $v(c_k^{\circ})= fs(k) >2\gamma$, 
the reduced chore cannot be allocated to $D_t$.
\end{proof}

Suppose that FFD allocates the reduced chore $c_k^{\circ}$ to a bundle $D_j$. We must have  $j\leq k$. Moreover, 
from Proposition \ref{lem-regular-bundles},
Claim \ref{claim-ak-2-regular} and 
Claim \ref{claim-notin-1regular},
it follows that $A_j$ must be 2-regular,
and moreover, all bundles $A_i$ for $i\in\{j,\ldots,k\}$ must be 2-regular.
We prove that the sums of costs of the two regular chores, in all bundles $\alloci[j],\ldots,\alloci[k]$, are equal. Please refer \Cref{fig:2-regular} for this situation. Formally:
\begin{claim}
\label{claim:shift}
If FFD allocates the reduced chore to bundle $A_j$, 
then for all $i\in\{j,\ldots,k\}$, 
$v(\{\alloci[i][1], \alloci[i][2]\})=v(\{\alloci[k][1], \alloci[k][2]\})$.
\end{claim}
\begin{proof}
By definition of regular chores, the cost of regular chores is not increasing with $i$.
So $v(\{\alloci[k][1], \alloci[k][2]\})\le v(\{\alloci[i][1], \alloci[i][2]\})\le v(\{\alloci[j][1], \alloci[j][2]\})$ for any $j\le i\le k$. 
So it is sufficient to prove that $v(\{\alloci[j][1], \alloci[j][2]\})
\leq v(\{\alloci[k][1], \alloci[k][2]\})$. 

By Claim \ref{claim-ak-2-regular}, the cost $v(c_k^{\circ})=fs(k)=\tau-v(\{\alloci[k][1], \alloci[k][2]\})$. As chores $\alloci[j][1]$ and $\alloci[j][2]$ are regular chores, they are larger than $c_k^{\circ}$, so their allocation  would not be influenced by the chore $c_k^{\circ}$. Therefore, we have $\alloci[j][1],\alloci[j][2]\in D_j$. 
Since FFD allocates $c_k^{\circ}$ to $D_j$,
$v(c_k^{\circ})\le \tau-v(\{\alloci[j][1], \alloci[j][2]\})$. 
So
$v(\{\alloci[j][1], \alloci[j][2]\})\le \tau-v(c_k^{\circ}) = v(\{\alloci[k][1], \alloci[k][2]\})$.
As $v(\{\alloci[k][1], \alloci[k][2]\})\le v(\{\alloci[j][1], \alloci[j][2]\})$, we have $v(\{\alloci[k][1], \alloci[k][2]\})=v(\{\alloci[j][1], \alloci[j][2]\})$.
\end{proof}


\begin{figure}
    \centering
    \includegraphics[width=16cm]{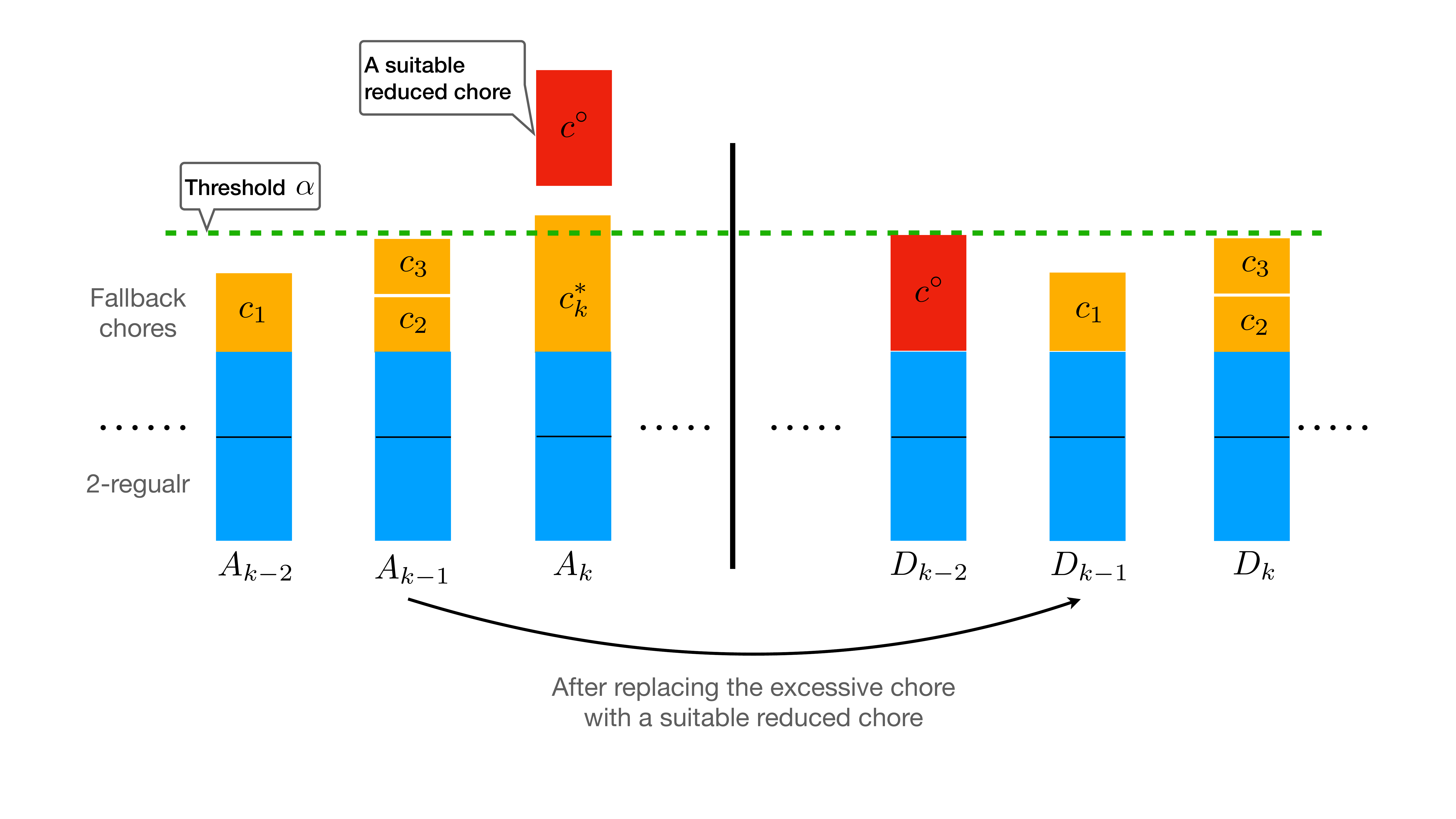}
    \caption{The shift of fallback chores after replacing the excessive chore, in the proof of Claim \ref{claim:shift} (for a large fit-in space).}
    \label{fig:2-regular}
\end{figure}

As $v(\{\alloci[k][1], \alloci[k][2]\})
=
v(\{\alloci[j][1],\alloci[j][2]\})$,
after FFD allocates $c_k^{\circ}$ to bundle $D_j$, 
we have
\begin{align*}
D_j=\{\alloci[j][1],\alloci[j][2],c_k^{\circ}\}
\end{align*}
so $v(D_j)=\tau$ and there is no room in $D_j$ for any other chore. 
This means that, when FFD constructs the bundle $D_{j+1}$, 
all the non-regular chores in $A_j$ are available. To see what happens in this case, refer to \Cref{fig:2-regular}.
FFD first puts into $D_{j+1}$ the two regular chores $A_{j+1}[1]$ and $A_{j+1}[2]$; then, the situation is exactly the same as when FFD constructed $A_j$ in the original run: the remaining space is the same (as $v(\{\alloci[j+1][1], \alloci[j+1][2]\})=v(\{\alloci[j][1], \alloci[j][2]\})$) and the set of remaining chores is the same (as all non-regular chores in $A_j$ are available). Therefore, 
\begin{align*}
D_{j+1} = \{\alloci[j+1][1],\alloci[j+1][2]\}\cup(\alloci[j]\setminus\{\alloci[j][1],\alloci[j][2]\}),
\end{align*}
that is, $D_{j+1}$ contains the regular chores of $A_{j+1}$ and the non-regular chores of $A_j$.
By similar arguments, a similar pattern occurs for all $i\in\{j+1,\ldots,k\}$:  $D_i=\{\alloci[i][1],\alloci[i][2]\}\cup(\alloci[i-1]\setminus\{\alloci[i-1][1],\alloci[i-1][2]\})$: the non-regular chores shift from the bundle with index $i-1$ to the bundle with index $i$. Finally, we have
\begin{align*}
D_{k} = \{\alloci[k][1],\alloci[k][2]\}\cup(\alloci[k-1]\setminus\{\alloci[k-1][1],\alloci[k-1][2]\}),
\end{align*}
The reason for this pattern is that when we have the same space left with same set of chores to compete, FFD will give the same allocation. 
so the equation in Definition \ref{def-suitable-redu} holds, and Lemma \ref{lem-suitable-with-large-fs}  is proved.
\end{proof}

Lemmas \ref{lem-suitable-with-small-fs},
\ref{lem-suitable-with-medium-fs} 
and
\ref{lem-suitable-with-large-fs} together imply Lemma \ref{lem-suitable-reduced-chore}.

\newpage
\section{Proof of \Cref{thm-ffd-is-monotone}}
\label{app:benchmark}
In this section we prove \Cref{thm-ffd-is-monotone} from Section \ref{sec:monotonicity}.

To simplify the proof, we introduce a new notation. Given any bundle $B$, we use  $B[\ge\gamma]$ to denote the set of chores in bundle $B$ with cost larger than $\gamma$. Formally, we have $B[\ge\gamma]=\{c\in B\mid v(c)\ge\gamma\}$.
\begin{lemma}
\label{lem-benchmark}
Given a First Fit Valid $(\items, \allocs, v, \tau)$ and any number $\gamma\ge0$, let $B_k$ be the $k$-th benchmark bundle of $\allocs$. We have the following:

(a) $\alloci[k][\ge\gamma]\lex{\ge}B_k[\ge\gamma]$.

(b) If $\alloci[k][\ge\gamma]\lex{>}\bundle_k[\ge\gamma]$, then $v(\alloci[k][\ge\gamma])>\tau$.
\end{lemma}

\begin{proof}
To make the proof more readable, let $\alloci[k]'=\alloci[k][\ge\gamma]$ and $B_k'=B_k[\ge\gamma]$.

\paragraph{(a)}
To prove $\alloci[k]'\lex{\ge}B_k'$, 
we suppose the contrary that $\alloci[k]'\lex{<}B_k'$. Then we prove that this implies $\alloci[k]\lex{<}B_k$. This would contradict the assumption that the tuple $(\items, \allocs, v, \tau)$ is First Fit Valid. 

Let $p$ be the smallest integer such that $v(\alloci[k]'[p])\neq v(B_k'[p])$.  When  $\alloci[k]'\lex{<}B_k'$, we have $v(\alloci[k]'[p])< v(B_k'[p])$. We have two cases here.

Case 1: $p\le |\alloci[k]'|$.
Then 
all $A_k'[p] = A_k[p]$, so 
$v(\alloci[k][p])<v(B_k[p])$, whereas $v(\alloci[k][q])=v(B_k[q])$ for all $q<p$. So $\alloci[k]\lex{<}B_k$.

Case 2: $p>|\alloci[k]'|$.
Then, 
for $q\le |A_k'|$, we have $v(\alloci[k][q])=v(B_k[q])$. 
For $|A_k'|<q\le p$, we have $v(\alloci[k][q])<v(B_k[q])$,
as $v(B_k[q])=v(B_k'[q])\ge\gamma$ whereas $v(\alloci[k][q])< \gamma$ by the fact that $q>|A_k'|$.
This implies $\alloci[k]\lex{<}B_k$. 

In both cases, we get a contradiction.
So we must have  $\alloci[k]'\lex{\ge}B_k'$.

\paragraph{(b)}
If $\alloci[k]'\lex{>}B_k'$, we prove that $\alloci[k]'\lex{>}B_k$. Let $p$ be the smallest integer such that $v(\alloci[k]'[p])\neq v(B_k'[p])$. We have $v(\alloci[k]'[p])>v(B_k'[p])$. 
For any $q\le p$, we have  $B_k[q]=B_k'[q]$ or $v(B_k[q]) < \gamma$. 
This implies
 $v(\alloci[k]'[p])>v(B_k[p])$ and $ v(\alloci[k]'[q])\ge v(B_k[q])$ for any $q\le p$. Therefore, $\alloci[k]'\lex{>}B_k$. 
Since $A_k'\subseteq A_k$,
by Lemma \ref{lem-lex-larger-implies-cost-larger} we have $v(\alloci[k]')>\tau$.
\end{proof}

Next, we prove a lemma on the case of two bundles.

\begin{lemma}\label{lem-2bundle}
Suppose $FFD(\items,v, \tau,2)$ allocates all chores into two bundles.
Given any First Fit Valid tuple $(\items,\allocs,v, \tau)$ such that $|\allocs|=2$, we have $\alloci[1]\cup\alloci[2]=\items$.
\end{lemma}
\begin{proof}
Let $\mathbf{D}:=FFD(\items,v, \tau,2)$. By the lemma assumption, we have $D_1\cup D_2=\items$. By Proposition \ref{prop-lexiffd}, the bundle $D_1$ is equal to $B_1$ --- the 1st benchmark bundle of $\mathbf{A}$.  Because $(\items,\allocs,v, \tau)$ is First Fit Valid, we have $\alloci[1]\lex{\ge}B_1=D_1$, so $v(\alloci[1])\ge v(D_1)$. So $v(\items\setminus\alloci[1])\le v(\items\setminus D_1)=v(D_2)\le \tau$. 
By First Fit Valid properties, the bundle $\alloci[2]$ must contain all the remaining chores. 
\end{proof}

\begin{proof}[Proof of \Cref{thm-ffd-is-monotone}]
Let $\mathbf{D}=FFD(\items,v, \tau)$ be the output of FFD algorithm. Suppose that $\mathbf{D}$ allocates all chores in the first $n$ bundles. Let tuple $(\items, \allocs, v, \tau)$ be First Fit Valid and $|\allocs|=n$.

The case $n=2$ immediately follows from \Cref{lem-2bundle}.

Next, we consider
the case $n=3$ and $\alpha\ge\frac{10}{9}$. 
Let $\tau := \alpha \mu$.
We are given a First-Fit-Valid tuple $(\items,\allocs,v, \tau)$, and have to prove that all chores are allocated.

By First-Fit-Valid properties,  $\alloci[1]\lex{\ge}B_1=D_1$, where $B_1$ is the benchmark bundle. If $\alloci[1]=D_1$, then we can apply \Cref{lem-2bundle} to chores set $\items\setminus D_1 = \items\setminus A_1$ and bundles $\alloci[2]$ and $\alloci[3]$:
since FFD$(\items\setminus D_1, v, \tau) = (D_2,D_3)$, and the tuple $(\items\setminus A_1, (A_2,A_3), v, \tau)$ is First Fit Valid, \Cref{lem-2bundle} implies that $\alloci[2]$ and $\alloci[3]$ should contain all the chores in set $\items\setminus D_1=\items\setminus\alloci[1]$.

Next, we consider the case $\alloci[1]\lex{>}D_1$.  
Since $D_1$ is the 1st benchmark bundle, which is maximal among bundles with cost at most $\tau$, we have $v(\alloci[1])>\tau\ge\frac{10}{9} \mu$. 
Suppose for contradiction that the bundle collection $\allocs$ does not contain all chores. 
By First Fit Valid properties this means that, without the first two bundles, the total remaining cost is more than $\tau$, that is,  $v(\items\setminus(\alloci[1]\cup\alloci[2]))>\tau \ge\frac{10}{9} \mu$,
By First Fit Valid properties this means that, without the first two bundles, the total remaining cost is more than $\tau$, that is,  $v(\items\setminus(\alloci[1]\cup\alloci[2]))>\tau \ge\frac{10}{9} \mu$,
Since the maximin share is $\mu$, $v(\items)\le 3z$, so
$v(\alloci[1]\cup\alloci[2]) < \frac{17}{9} \mu$,
so $v(\alloci[2])<\frac{7}{9} \mu$. As there is an unallocated chore $c^*$, it cannot be allocated to $\alloci[2]$, so $v(c^*) > \tau-\frac{7}{9}\mu\geq \frac{1}{3} \mu$.


As small cost chores cannot influence the allocation of large cost chores, we can forget all chores with cost less than $v(c^*)$.
Let us consider the set $\items'=\items[\geq v(c^*)]$,
which only contains the chores with cost at least as high as $c^*$.
Similarly, define  $D_k'=D_k[\geq v(c^*)]$ and $\alloci[k]'=\alloci[k][\geq v(c^*)]$.  Notice that $\mathbf{D}'=FFD(\items',v, \tau)$. By \Cref{lem-benchmark}(a), we have $\alloci[1]'\lex{\ge}D_1'$. If $\alloci[1]'\lex{=}D_1'$, then we can apply \Cref{lem-2bundle} on chores set $\items\setminus D_1'$ and bundles $\alloci[2]'$ and $\alloci[3]'$. Then we have $\alloci[2]'\cup\alloci[3]'=\items'\setminus\alloci[1]'=\items'\setminus D_1'$. There are no remaining chores.

The only case left is $\alloci[1]'\lex{>}D_1'$. By \Cref{lem-benchmark}(b), we have $v(\alloci[1]')>\tau\ge\frac{10}{9}\mu$. So $|\alloci[1]'|\ge2$.  
On the other hand, notice that $|\items'|\le 6$, since MMS$(\items',v,3)\leq \mu$, so there must exist a partition of $\items'$ in which each bundle contains at most 2 chores with cost greater than $\frac{1}{3}\mu$. Therefore, we have $|\items'\setminus\alloci[1]'|\le 4$.

As at least one chore ($c^*$) remains unallocated in $\allocs'$, we have $|A_2' \cup A_3'|\leq 3$, so one of $A_2',A_3'$ must contain a single chore; suppose it is $A_2'$, and denote the single chore by $c_2$. Since $c^*$ cannot be allocated to $A_2'$, we have $v(c_2)+v(c^*)>\tau>\mu$. As all chores in $\items'$ have a cost of at least $v(c^*)$,
chore $c_2$ must be in a singleton bundle in any MMS partition.
Moreover, $c_2$ is not the largest chore in $\items'$, as the largest chore (denote it by $c_1$) is allocated to $A_1'$ by First-Fit-Valid properties. Therefore, $c_1$ too must be in a singleton bundle in any MMS partition. As any bundle in an MMS partition contains at most two chores, we get $|\items'|\le 4$, so $|\items'\setminus\alloci[1]'|\le 2$. This is a contradiction, as two chores can always be allocated into two bundles.

\end{proof}

\end{document}

%% file: notaion.tex
\newcommand{\inst}{\mathcal{I}}

\newcommand{\bundle}{B} 
\newcommand{\bundlei}[1][i]{\bundle_{#1}} 

\newcommand{\alloc}{A} 
\newcommand{\allocs}{\mathbf \alloc} 
\newcommand{\alloci}[1][i]{\alloc_{#1}} 

\newcommand{\agents}{\mathcal{N}} 
\newcommand{\items}{\mathcal{M}} 

\newcommand{\vai}[2][i]{v_{#1}\left(#2\right)} 

\newcommand{\last}{\omega}

\newcommand{\lex}[1]{\stackrel{\mathclap{\tiny lex}}{#1}}